\documentclass[11pt,times,letter]{article}

\usepackage{amsmath, amsfonts,amsthm,amssymb,bm}
\usepackage{dsfont}

\usepackage[english]{babel}

\usepackage[normalem]{ulem}

\usepackage{latexsym,amssymb,amsfonts,graphicx}
\usepackage{amsmath,dsfont}
\usepackage{verbatim}
\usepackage{mathrsfs}
\usepackage{bm}
\usepackage{color}
\usepackage{epsfig}
\usepackage{tabularx}
\usepackage{url}
\usepackage{bm}
\usepackage{natbib}
\usepackage{epstopdf}
\usepackage{algorithm}
\usepackage{algpseudocode}

\usepackage[colorlinks,linkcolor=blue,citecolor=blue,anchorcolor=blue]{hyperref}
\usepackage{subfigure}

\newcommand{\Px}{ \mathbb{P} }
\newcommand{\Qx}{ \mathbb{Q} }

\newcommand{\Ex}{ \mathbb{E} }

\def\esssup_#1{\underset{#1}{\mathrm{ess\,sup\, }}}
\def\essinf_#1{\underset{#1}{\mathrm{ess\,inf\, }}}
\def\argmax_#1{\underset{#1}{\mathrm{arg\,max\, }}}
\def\argmin_#1{\underset{#1}{\mathrm{arg\,min\, }}}

\newcommand{\Fx}{\mathbb{F} }

\newcommand{\F}{\mathcal{F}}

\newcommand{\R}{\mathds{R}}

\newtheorem{theorem}{Theorem}[section]
\newtheorem{definition}{Definition}[section]
\numberwithin{equation}{section}

\newtheorem{proposition}[theorem]{Proposition}
\newtheorem{remark}[theorem]{Remark}
\newtheorem{lemma}[theorem]{Lemma}

\allowdisplaybreaks[4]

\definecolor{Red}{rgb}{1.00, 0.00, 0.00}

\definecolor{DRed}{rgb}{0.5, 0.00, 0.00}

\definecolor{Blue}{rgb}{0.00, 0.00, 1.00}

\definecolor{Green}{rgb}{0.0, 0.4, 0.0}

\title{On optimal tracking portfolio in incomplete markets: The reinforcement learning approach}

\author{Lijun Bo \thanks{Email: lijunbo@ustc.edu.cn, School of Mathematics and Statistics, Xidian University, Xi'an, 710126, China.}
\and
Yijie Huang \thanks{Email: huang1@mail.ustc.edu.cn, School of Mathematical Sciences, University of Science and Technology of China, Hefei, 230026, China.}
\and
Xiang Yu \thanks{Email: xiang.yu@polyu.edu.hk, Department of Applied Mathematics, The Hong Kong Polytechnic University, Hung Hom, Kowloon, Hong Kong.}
}
\date{\vspace{-1.5cm}}

\begin{document}
\maketitle

\begin{abstract}
This paper studies an infinite horizon optimal  tracking portfolio problem using capital injection in incomplete market models. The benchmark process is modelled by a geometric Brownian motion with zero drift driven by some unhedgeable risk. The relaxed tracking formulation is adopted where the fund account compensated by the injected capital needs to outperform the benchmark process at any time, and the goal is to minimize the cost of the discounted total capital injection. When model parameters are known, we formulate the equivalent auxiliary control problem with reflected state dynamics, for which the classical solution of the HJB equation with Neumann boundary condition is obtained explicitly. When model parameters are unknown, we introduce the exploratory formulation for the auxiliary control problem with entropy regularization and develop the continuous-time q-learning algorithm in models of reflected diffusion processes. In some illustrative numerical example, we show the satisfactory performance of the q-learning algorithm.

\vspace{0.1in}
\noindent{\textbf{Keywords}}: Optimal tracking portfolio, capital injection, incomplete market, reflected diffusion process, continuous-time q-learning\\

\noindent{\textbf{MSC 2020}}: 93E20, 49J55,	68T05 
\end{abstract}

\section{Introduction}
One important business in fund management is to choose the portfolio among some risky assets to closely track some benchmark processes such as the market index, the inflation rates, the exchange rates, the liability, or the living cost and education cost, etc. How to formulate the tracking procedure and derive the optimal tracking portfolio has become an active research topic in quantitative finance during the past decades. For example, several objectives for active portfolio management are proposed in \cite{Browne9}, \cite{Browne99}, and \cite{Browne00} including: maximizing the probability that the agent's wealth achieves a performance goal related to the benchmark before falling below a predetermined shortfall; minimizing the expected time to reach the performance goal; the mixture of these two objectives, and some further extensions by considering the expected reward or expected penalty. Another standard method to optimize the tracking error is to minimize the variance or downside variance relatively to the index value or return, which leads to the linear-quadratic stochastic control problem; See \cite{Gaivoronski05}, \cite{YaoZZ06} and \cite{NLFC} among others. In \cite{Strub18}, a different objective function is introduced to measure the similarity between the normalized historical trajectories of the tracking portfolio and the index where the rebalancing transaction costs can also be taken into account. Recently, another new tracking formulation using the fictitious capital injection is studied in \cite{BoLiaoYu21,BHY23a, BHY23b}, however, all in the complete market model. With the help of an auxiliary control problem based on the state process with reflections, the original stochastic control problem under dynamic floor constraints is equivalent to the study of the HJB equation with a Neumann boundary condition. The market completeness plays a key role therein as the dual transform of the HJB equation leads to a linear PDE problem whose solution admits some probabilistic representations. The existence of the classical solution for the linear dual PDE and the verification of the feedback optimal portfolio control can be obtained by using stochastic flow analysis and some delicate estimations.

The present paper aims to further extend the studies on the relaxed tracking portfolio using capital injection in \cite{BoLiaoYu21,BHY23a,BHY23b} in some incomplete market models in which there exists some unhedgeable risk driving the external benchmark process (see its definition in \eqref{eq:W-kappa}). As a direct consequence, the previous methodology based on the dual transform and the probabilistic representations for the dual PDE therein fails because the dual HJB equation can no longer be linearized. We consider the benchmark process governed by a geometric Brownian motion with zero drift in the present paper. With knowledge of market parameters, we formulate an equivalent auxiliary stochastic control problem similar to \cite{BoLiaoYu21}, which results in a stochastic control problem with the underlying state process exhibiting reflections at
the boundary zero. Leveraging the benchmark dynamics as a geometric Brownian motion (GBM) facilitates the dimension reduction of the control problem. We then derive the explicit classical solution to the associated HJB equation with a Neumann boundary condition. Furthermore, we consider the realistic situation when all market model parameters are unknown, for which we are interested in developing the continuous-time reinforcement learning approach. Recently, for continuous-time stochastic control problems, \cite{Wang20}, \cite{JZ22a,JZ22b,JZ22c} have developed the theoretical foundation for reinforcement learning in the continuous time exploration formulation with continuous state space and action space. In particular, \cite{Wang20} firstly studied the exploratory learning formulation by incorporating the entropy regularization to encourage the policy exploration, and the optimal policy has been shown as a Gaussian type for LQ control problems. \cite{JZ22a} examined the policy evaluation problem and proposed the martingale loss to facilitate the algorithm design. The policy gradient problem was then studied in \cite{JZ22b} where the martingale approach in \cite{JZ22a} can be adopted for the policy gradient and some Actor-Critic algorithms can be devised to learn the value function and the stochastic policies alternatively. Later, \cite{JZ22c} established a continuous time q-learning theory by considering the first order approximation of the conventional discrete-time Q-function. We also refer to some recent subsequent studies in different contexts such as \cite{W23} in applying the continuous-time Actor-Critic algorithm for solving the optimal execution problem in the Almgren-Chriss model; the application of policy improvement algorithm in \cite{BGMX} for solving the optimal dividend problem in diffusion models; the generalization of continuous-time q-learning in \cite{weiyu2023} for mean-field control problems in McKean-Vlasov diffusion models. The occupation measure considered in \cite{ZhaoTD2024} facilitates the derivation of performance-difference and local-approximation formulas for continuous-time policy gradient (PG) and trust region policy optimization (TRPO) algorithms. 

In this paper, we are interested in generalizing the continuous-time q-learning from the standard diffusion models in \cite{JZ22c} to solve our auxiliary stochastic control problem with the reflection boundary at $0$. In particular, for the stochastic control problems with reflections, it is shown in the present paper that the value function and the associated q-function can also be characterized by the martingale condition of some stochastic processes involving the reflection term or the local time at $0$; see Proposition \ref{prop:martingale-condition} and Theorem \ref{thm:martingale-condition}. In addition, some extra transversality conditions (see \eqref{eq:cindition-q} and \eqref{eq:condition-pi}) are also needed in our infinite horizon setting. 
Building upon our martingale characterization in Theorem \ref{thm:martingale-condition}, we are able to devise the offline q-learning algorithm for the targeted stochastic control problems with reflections using the stochastic approximation resulting from the martingale orthogonality condition. As a new theoretical contribution to the literature, in the context of the infinite horizon stochastic control problems with reflections, we establish the convergence of the stochastic approximation algorithms for time discretization and horizon truncation when $\Delta t\rightarrow 0$ and $T\rightarrow\infty$; see Theorem \ref{thm:convergence-policy}. To illustrate the efficiency of our q-learning algorithm, we study a simulation example of our optimal tracking portfolio problem with a particular choice of the temperature parameter $\gamma=\rho/d$. Here $\rho$ stands for the subjective discount rate and $d$ is the dimension of stocks, both are known to the agent. In this case, we present an example beyond the LQ control framework such that the exploratory HJB equation with a Neumann boundary condition can be solved fully explicitly, allowing us to obtain the exact parameterization of the optimal value function and the optimal q-function. For some initial inputs and choices of learning rates, we illustrate the very satisfactory performance of the iteration convergence of learned parameters towards their true values.

The remainder of this paper is organized as follows. Section \ref{sec:formulation} introduces the incomplete market model and the benchmark process as well as the portfolio optimization problem under the relaxed tracking formulation. Section \ref{sec:PDE} delves into the classical control problem when the model is known. By adopting an auxiliary control formulation and conducting dimension reduction, the HJB equation is solved explicitly. In Section \ref{sec:RL}, when the market parameters are unknown, the continuous-time reinforcement learning approach is proposed for the exploratory formulation of the control problem with entropy regularization. In particular, the q-learning algorithm is developed therein together with some convergence results on time discretization and horizon truncation. Section \ref{sec:example} studies some numerical examples by implementing the q-learning algorithm. In addition, we present the comparison between the q-learning algorithm and the maximum likelihood estimation (MLE) for some real market data to demonstrate the outperformance of our algorithm. Finally, all proofs of the main results in previous sections are collected in Section \ref{appendix:proof}.

\section{Market Model and Problem Formulation}\label{sec:formulation}

Let $(\Omega, \mathcal{F}, \Fx^{W},\mathbb{P}^{W})$ be a filtered probability space with the filtration $\mathbb{F}^W=(\mathcal{F}^W_t)_{t\geq 0}$ satisfying the usual conditions, which supports a $d+1$-dimensional Brownian motion $(W^0,W^1,\ldots,W^d)=(W_t^0,W_t^1,\ldots,W_t^d)_{t\geq 0}$. We consider a financial market consisting of $d$ risky assets whose price dynamics are driven by the d-dimensional Brownian motion $(W^1,\ldots, W^d)$ satisfying
\begin{align}\label{stockSDE}
\frac{dS_t^i}{S_t^i}= \mu_idt+\sum_{j=1}^d\sigma_{ij}dW_t^{j},\quad i=1,\ldots,d,\quad t\geq 0,
\end{align}
with the mean return rate $\mu_i\in\R$ and the volatility $\sigma_{ij}\in\R$. Hereafter, we denote by  $\mu=(\mu_1,\ldots,\mu_d)^{\top}$ the vector of return rate and by $\sigma=(\sigma_{ij})_{d\times d}$ the volatility matrix. Assume that the riskless interest rate $r=0$ that amounts to the change of num\'{e}raire and the return rate $\mu\neq {\bm 0}$. From this point onwards, all processes including the wealth process and the benchmark process are defined after the change of num\'{e}raire.

For $t\geq 0$, let $\theta_t^i$ be the amount of wealth (the resulting process $\theta^i=(\theta_t^i)_{t\geq 0}$ is assumed to be $\Fx^W$-adapted) that the fund manager allocates in asset $S^i=(S_t^i)_{t\geq 0}$ at time $t$. The self-financing wealth process under the portfolio control $\theta=(\theta_t^1,\ldots,\theta_t^d)_{t\geq 0}^{\top}$ is given by, for $t\geq 0$,
\begin{align}\label{eq:wealth2}
V^{\theta}_t &=\textrm{v}+\int_0^t\theta_s^{\top}\mu ds+\int_0^t\theta_s^{\top}\sigma dW_s,\quad V_0^{\theta}=\textrm{v}\in\R_+:=[0,\infty).
\end{align}

 We consider the portfolio decision making by the fund manager whose goal is to optimally track a stochastic benchmark process $Z=(Z_t)_{t\geq 0}$. In the present paper, we consider the benchmark process satisfying the  geometric Brownian motion (GBM) with zero drift as studied in \cite{BJLand2000} that may refer to futures price process of a commodity, a stock index, or an interest rate index. As in \cite{BJLand2000}, the benchmark process $Z=(Z_t)_{t\geq 0}$ is assumed to satisfy the form of
\begin{align}\label{factor-Z}
\frac{dZ_t}{Z_t} =\sigma_ZdW^{\kappa}_t,\quad \forall t>0,
\end{align}
with the initial value $Z_0=z>0$ and the volatility $\sigma_Z>0$.

For the nonzero correlative coefficient $\kappa\in [-1,1]$, the process $W^{\kappa}=(W_t^{\kappa})_{t\geq 0}$ appeared in \eqref{factor-Z} is a standard Brownian motion, which is defined by
\begin{align}\label{eq:W-kappa}
W^{\kappa}_t:=\kappa W^0_t+\sqrt{1-\kappa^2}W^{\eta}_t,\quad t\geq 0.
\end{align} 
Here $W^{\eta}=(W^{\eta}_t)_{t\geq 0}$ is a linear combination of $d$-dimensional Brownian motion $W=(W^1,\ldots, W^d)$ with weights $\eta=(\eta_1,\ldots,\eta_d)^{\top}\in[-1,1]^d$. We also recall that $W^0=(W^0_t)_{t \geq 0}$ is a Brownian motion independent of the $d$-dimensional Brownian motion $W$, which stands for the unhedgeable risk. As a consequence, the market model is incomplete. 

Given the benchmark process $Z=(Z_t)_{t\geq0}$, an optimal tracking portfolio problem is considered that combines the portfolio control with another singular control as capital injection together with dynamic floor constraints. To be precise, we assume that the fund manager can strategically inject capital  
$A=(A_t)_{t\geq 0}$ into the fund account from time to time whenever it is necessary such that the total capital dynamically dominates the benchmark process $Z$, i.e., $A_t+V_t^{\theta}\geq Z_t$ for all $t\geq 0$. The goal of the optimal tracking problem is to minimize the expected cost of the discounted total capital injection under dynamic floor constraints that
\begin{align}\label{eq_prob_IBP}
{\rm w}(\mathrm{v}, z):=\text{$\inf_{A,\theta} \Ex\left[A_0+\int_0^{\infty} e^{-\rho t}dA_t \right]$\quad s.t.\quad $Z_t \le A_t + V^{\theta}_t$ at each $t\geq 0$,}
\end{align}
where the constant $\rho>0$ is the discount rate and $A_0=(z-\textrm{v})^+$ is the initial injected capital to match with the initial benchmark.

To cope with the problem \eqref{eq_prob_IBP} with dynamic floor constraints, we observe that, for a fixed control $\theta$, the optimal $A$ is the smallest adapted right-continuous and non-decreasing process that dominates $Z-V^{\theta}$. Let $\mathbb{U}$ be the set of regular $\Fx$-adapted control processes $\theta=(\theta_t)_{t\geq 0}$ such that \eqref{eq:wealth2} is well-defined. Similar to \cite{BoLiaoYu21}, for each fixed regular control $\theta$, the optimal singular control $A_t^*$ satisfies that
$A^*_t =0\vee \sup_{s\in[0,t]}(Z_s-V_s^{\theta})$, $\forall t\geq 0$.
As a result, the problem \eqref{eq_prob_IBP} with floor constraints admits the equivalent formulation as a unconstrained control problem but with a running maximum cost that
\begin{equation}\label{eq_orig_pb}
{\rm w}(\mathrm{v}, z)=(z-\mathrm{v})^++\inf_{\theta\in\mathbb{U}}\ \Ex\left[ \int_0^{\infty} e^{-\rho t} d\left(0\vee \sup_{s\in[0,t]}(Z_s-V_s^{\theta})\right)\right].
\end{equation}
In the forthcoming section, we will focus on the solvability of problem \eqref{eq_orig_pb} in the sense of strong control by introducing an auxiliary control problem.  

\section{The Auxiliary Control Problem}\label{sec:PDE}

To formulate the auxiliary control problem, we impose a new controlled state process to replace the process $V^{\theta}=(V_t^{\theta})_{t\geq 0}$ given in \eqref{eq:wealth2}. To do it, we first define the following difference process by
%\begin{align}\label{eq:Dt}
$D_t:=Z_t-V_t^{\theta}+\mathrm{v}-z$,  $\forall t\geq 0$.
%\end{align}
It is obvious that $D_0=0$. Moreover, for any $x\geq 0$, let us consider the running maximum process of $D=(D_t)_{t\geq 0}$ defined by
\begin{align}\label{eq:maxM}
L_t :=x\vee \sup_{s\in[0,t]}D_s\geq 0,\quad \forall t>0,
\end{align}
with the initial value $L_0=x\geq 0$. One can easily see that $(z-\mathrm{v})^+-{\rm w}(z,\mathrm{v})$ with the value function ${\rm w}(\mathrm{v},z)$ given in \eqref{eq_orig_pb} is equivalent to the auxiliary control problem:
\begin{align}\label{eq:objectivefcn}
\sup_{\theta\in\mathbb{U}}\Ex\left[-\int_0^{\infty}e^{-\rho t}dL_t\right],
\end{align}
when we set the initial level $L_0=x=(\mathrm{v}-z)^+$. 

We start with the introduction of a new controlled state process $X=(X_t)_{t\geq 0}$ for problem \eqref{eq:objectivefcn}, which is defined as the reflected process $X_t:=L_t-D_t$ for $t\geq 0$ that satisfies the SDE:
\begin{align}\label{state-X}
X_t = -\int_0^t \sigma_Z Z_s dW^{\kappa}_s +\int_0^t\theta_s^{\top}\mu ds+\int_0^t\theta_s^{\top}\sigma dW_s + L_t,\quad \forall t\geq0
\end{align}
with the initial value $X_0=x\geq 0$. In particular, the running maximum process $L_t$ increases if and only if $X_t=0$, i.e., $L_t=D_t$. We will change the notation from $L_t$ to $L_t^X$ from this point onwards to emphasize its dependence on the new state process $X$ given in \eqref{state-X}. The benchmark process $Z=(Z_t)_{t\geq 0}$ defined in \eqref{factor-Z}  is chosen as another state process.

Let $\mathbb{U}^r$ be the set of admissible portfolios (controls) such that the reflected SDE \eqref{state-X} has a unique strong solution. Then, we propose the following stochastic control problem, for $(x,z)\in\R_+\times (0,\infty)$,
\begin{align}\label{eq:hatw}
\hat{{\rm w}}(x,z):=\sup_{\theta\in\mathbb{U}^r}J(\theta;x,z):=\sup_{\theta\in\mathbb{U}^r}\Ex_{x,z}\left[-\int_0^{\infty}e^{-\rho t}dL_t^X\right],
\end{align}
where $\Ex_{x,z}[\cdot]:=\Ex[\cdot|X_0=x,Z_0=z]$. Using the definition \eqref{eq:hatw} of $\hat{\rm w}(x,z)$, it is not difficult to check that (i) the value function $x\to \hat{{\rm w}}(x,z)$ is non-decreasing; (ii) it holds that $\sup_{x_1,x_2\geq0}|\frac{\hat{{\rm w}}(x_1,z)-\hat{{\rm w}}(x_2,z)}{x_1-x_2}|\leq1$, that is, the value function $x\to \hat{{\rm w}}(x,z)$ is Lipschitz continuous with Lipschitz coefficient being $1$. These properties will be applied in the remaining sections.

We next apply the change of measure to simplify the stochastic control problem in  \eqref{eq:hatw} by reducing the dimension. To this end, let us introduce the normalized processes of triplet $(X,\theta,L^X)$  by the benchmark process $Z=(Z_t)_{t\geq0}$ defined  by
\begin{align}\label{eq:Y}
    &Y_t:=\frac{X_t}{Z_t},\quad \tilde{\theta}_t:=\frac{\theta_t}{Z_t},\quad L_t^Y:=\int_0^t \frac{dL_s^X}{Z_s},\quad \forall t\geq0.
\end{align}
Note that the process $Z$ is strictly positive (see \eqref{factor-Z}). Then, $t\to L_t^Y$ is a non-decreasing process and satisfies that, a.s., $\int_{0}^t{\bf1}_{Y_s=0}dL_s^Y=\int_{0}^t{\bf1}_{X_s=0}\frac{dL_s^X}{Z_s}=\int_0^t\frac{dL_s^X}{Z_s}=L_t^{Y}$ for $t\geq 0$. This implies that $L^Y=(L_t^Y)_{t\geq 0}$ is the local time process for the process $Y=(Y_t)_{t\geq 0}$ at the reflecting boundary $0$. Then, the It\^o's rule yields that
\begin{align}\label{eq:SDEY}
dY_t&=\sigma_Z^2(Y_t+1)dt-\sigma_Z(Y_t+1)dW_t^{\kappa} +\tilde{\theta}_t^{\top}\mu dt + \tilde{\theta}_t^{\top}\sigma dW_t+ dL_t^Y.
\end{align}
In the sequel, we introduce the following change of measure specified by 
\begin{align}\label{eq:Q}
    \frac{d\Qx^W}{d\Px^W}\Big|_{\F_t}=Z_t,\quad \forall t\geq 0.
\end{align}
We then consider the following stochastic control problem formulated as:
\begin{align}\label{eq:stocontrolY00}
u(y)&:=\sup _{\theta \in \tilde{\mathbb{U}}^r}\Ex^{\Qx^W}\left[-\int_0^{\infty}e^{-\rho s}dL_s^Y\right],\quad \forall y\in\R_+,
\end{align}
where the normalized admissible set $\tilde{\mathbb{U}}^r:=\{(\frac{\theta_t}{Z_t})_{t\geq 0};~\theta\in\mathbb{U}^r\}$, and $\Ex^{\Qx}[\cdot]$ is the  expectation under the probability measure $\Qx^W\sim\Px^W$. It follows from \eqref{eq:SDEY} that the underlying state process $Y$ has the following dynamics under $\Qx^W$ given by
\begin{align}\label{eq:dynamcsY2Q}
dY_t&=-\sigma_Z(Y_t+1)d\tilde{W}^{\kappa}_t+\tilde{\theta}^{\top}_t\mu dt+\tilde{\theta}^{\top}_t\sigma dW_t+dL_t^Y.
\end{align}
Here $\tilde{W}_t^{\kappa}:=W^{\kappa}_t-\sigma_Z t$ for $t\geq 0$ is a $\Qx^W$-Brownian motion. For the case $\kappa=1$, the process $\tilde{W}^{\kappa}=(\tilde{W}_t^{\kappa})_{t\geq0}$ is also independent of $W=(W^1,\ldots,W^d)^{\top}$. The following lemma gives the relationship between the value function $u(y)$ for $y\geq0$ defined by \eqref{eq:stocontrolY00} and the value function $\hat{\rm w}(x,z)$ for $(x,z)\in\R_+\times(0,\infty)$ defined by \eqref{eq:hatw}, whose proof is straightforward and hence omitted.
\begin{lemma}\label{lem:reluv}
Let $(x,z)\in\R_+\times (0,\infty)$. It holds that $\hat{\rm w}(x,z)=zu\left(\frac{x}{z}\right)$.
\end{lemma}

Lemma~\ref{lem:reluv} tells us that we can solve the simplified stochastic control problem \eqref{eq:stocontrolY00} with one-dimensional state process with reflection instead of problem \eqref{eq:hatw}.  Then, the corresponding value function satisfies the following  HJB equation given by, for $y>0$,
\begin{align}\label{eq:HJB-uYY-infinite}
&\sup_{\tilde{\theta}\in\R^n}\left\{u'(y)\tilde{\theta}^{\top}\mu +u''(y)\tilde{\theta}^{\top}\sigma\sigma^{\top}\tilde{\theta}-\sqrt{1-\kappa^2}\sigma_Z(y+1)\tilde{\theta}^{\top}\sigma\eta u''(y)\right\}\nonumber\\
&\qquad+\frac{1}{2}\sigma_Z^2(y+1)^2u''(y)=\rho u(y)
\end{align}
with Neumann boundary condition $u'(0)=1$. If the value function $u$ is strictly concave, then the optimal feedback portfolio is given by 
\begin{align}\label{eq:optimal-conditon-u-infinite}
\tilde{\theta}^{*}(y)=-(\sigma \sigma^{\top})^{-1}\frac{u'(y) \mu-\sqrt{1-\kappa^2}\sigma_Z(y+1)\sigma\eta u''(y)}{u''(y)}, \quad\forall y\geq0.
\end{align}
By plugging \eqref{eq:optimal-conditon-u-infinite} into \eqref{eq:HJB-uYY-infinite}, we have
\begin{align}\label{eq:HJB-u-infinite}
-\alpha \frac{(u'(y))^2}{u''(y)}+\sqrt{1-\kappa^2}\zeta(y+1)u'(y)+\frac{1}{2}\sigma_Z^2\kappa^2(y+1)^2u''(y)=\rho u(y),
\end{align}
where the coefficients are defined by $\alpha:=\frac{1}{2}\mu^{\top}(\sigma\sigma^{\top})^{-1}\mu$ and $\zeta:=\sigma_Z \eta^{\top}\sigma^{-1}\mu$. We consider the  candidate solution of Eq.~\eqref{eq:HJB-u-infinite} given by
\begin{align}\label{eq:sol-u}
    u(y)=\frac{\lambda-1}{\lambda}(1+y)^{\frac{\lambda}{\lambda-1}},\quad \forall y\geq0. 
\end{align}
Here, $\lambda\in(0,1)$ is the unique solution to the following equation:
\begin{align}\label{eq:lambda}
\ell(\lambda):=\alpha \lambda(\lambda-1)^2+\rho(\lambda-1)^2-\sqrt{1-\kappa^2}\zeta\lambda(\lambda-1)-\frac{1}{2}\kappa^2\sigma_Z^2\lambda=0.
\end{align}
Hence, Eq.~\eqref{eq:lambda} has a unique root $\lambda\in(0,1)$. Based on the above results, we have the following verification result for the strong control problem, whose proof is standard and hence omitted. 
\begin{proposition}\label{thm:verification-strong}
The function  $u(y)$ for $y\geq0$ given by \eqref{eq:sol-u} is a classical solution of the HJB equation \eqref{eq:HJB-uYY-infinite}. Define the following optimal feedback control function by, for all $y\in\R_+$,
\begin{align}\label{eq:theta}
\tilde{\theta}^*(y)=
(1-\lambda)(y+1)(\sigma \sigma^{\top})^{-1}\mu+\sqrt{1-\kappa^2}\sigma_Z(y+1)(\sigma \sigma^{\top})^{-1}\sigma\eta.
\end{align}
Consider the controlled state process $Y^*=(Y_t^*)_{t \geq 0}$ that obeys the following SDE, for all $t \geq 0$,
\begin{align}\label{eq:optimal-Y}
Y_t^*=-\int_0^t\sigma_Z(Y_s+1)d\tilde{W}^{\kappa}_s+\int_0^t \tilde{\theta}^*(Y_s^*)^{\top}\mu ds+\int_0^t \tilde{\theta}^*(Y_s^*)^{\top}\sigma dW_s+L_t^{Y^*},
\end{align}
with $L_0^{Y^*}=y$. Let $\tilde{\theta}_t^*=\tilde{\theta}^*(Y_t^* )$ for all $t \geq 0$. Then $\tilde{\theta}^*=(\tilde{\theta}_t^*)_{t \geq 0} \in \tilde{\mathbb{U}}^{\mathrm{r}}$ is an optimal investment strategy. 
\end{proposition}

\section{The Continuous-time Reinforcement Learning Approach}\label{sec:RL}

In this section, the model is assumed to be unknown to the fund manager, i.e., all model parameters $\mu=(\mu_1,\ldots,\mu_d)^{\top}$, $\sigma=(\sigma_{ij})_{d\times d}$, $\sigma_Z$, $\kappa$ and $\eta=(\eta_1,\ldots,\eta_d)^{\top}\in[-1,1]^d$ are unknown. Thereby, the classical control approach in the previous section is not applicable. Our goal is to develop a continuous-time reinforcement learning algorithm to find the optimal tracking portfolio in problem \eqref{eq:stocontrolY00}. 

Reinforcement learning allows the decision maker to learn the optimal action in the unknown environment through the interactions with the environment as the repeated trial-and-error procedure. Specifically, the agent exercises a sequel of actions $\tilde{\theta}=(\tilde{\theta}_t)_{t\geq 0}$ and observe responses from $(Y^{\tilde{\theta}},L^Y)=(Y^{\tilde{\theta}}_t,L_t^Y)_{t\geq 0}$ along with a stream of discounted running rewards $(-e^{-\rho t}dL_t^{Y})_{t\geq 0}$, and continuously update and improve his actions based on these observations.

 To describe the exploration step in reinforcement learning, we randomize the action $\tilde{\theta}$ and consider its distribution. Assume that the probability space is rich enough to support the uniformly distributed random variable on $[0,1]$ that is independent of $(\tilde{W}^{\kappa},W^1,\ldots,W^d)$ and can be used to generate other random variables with specified density functions. Let $K=(K_t)_{t\in[0,T]}$ be a process of mutually independent copies of a uniform random variable on $[0,1]$ which is also independent of the BMs $(W^0,W^1,\ldots,W^d)$, the construction of which requires a suitable extension of probability space (c.f. \citealt{Sun06}). We then further expand the filtered probability space to $(\Omega, \mathcal{F},\Fx, \mathbb{Q})$ where $\Fx=(\mathcal{F}_t^W \vee \sigma\left(K_s, 0 \leq s \leq t\right))_{t\geq 0}$ and the probability measure $\Qx$, now defined on $\Fx$, is an extension from $\Qx^W$ (i.e. the two probability measures coincide when restricted to $\Fx^W$). In particular, let $\bm{\pi}:y \in \R_+\to \bm{\pi}(\cdot \mid y) \in \mathcal{P}({\cal A})$ be a given (feedback) policy with ${\cal A}:=\R^d$, and $\mathcal{P}({\cal A})$ is a suitable collection of probability density functions. At each time $t\geq0$, an action $\tilde{\theta}^{\bm \pi}_t$  is generated from the joint density $\bm{\pi}\left(\cdot \mid Y_t\right)$.  Fix a policy $\bm{\pi}$ and an initial state $y\geq0$, we can consider the following reflected SDE:
\begin{align}\label{eq:sample-Y}
dY_t^{\bm{\pi}}&=-\sigma_Z(Y_t^{\bm{\pi}}+1)d\tilde{W}^{\kappa}_t+(\tilde{\theta}^{\bm \pi}_t)^{\top}\mu dt+(\tilde{\theta}^{\bm \pi}_t)^{\top}\sigma dW_t+dL_t^{Y^{\bm{\pi}}},~t> 0; \quad Y^{\bm{\pi}}_0=y.
\end{align}
Here, $L^{Y^{\bm{\pi}}}=(L_t^{Y^{\bm{\pi}}})_{t\geq0}$ denotes the local time of the state process $Y^{\bm{\pi}}=(Y_t^{\bm{\pi}})_{t\geq0}$ at the level $0$.

To encourage exploration, we adopt the Shannon's entropy regularizer suggested in \cite{Wang20} that leads to the following value function associated to a given policy $\bm{\pi}$ that
\begin{align}\label{eq:obj-fuc}
J(y ; \bm{\pi})= \mathbb{E}^{\mathbb{Q}}\left[-\gamma\int_0^\infty  e^{-\rho t}\ln {\bm{\pi}}\left(\tilde{\theta}_t^{{\bm \pi}} \mid Y^{\bm{\pi}}_t\right) \mathrm{d} t-\int_0^\infty e^{-\rho t}dL_t^{Y^{\bm{\pi}}}\Big| Y_0^{\bm{\pi}}=y\right],
\end{align}
where $\mathbb{E}^{\mathbb{Q}}$ is the expectation with respect to both the Brownian motion and the action randomization. In the above, $\gamma > 0$ is a given parameter indicating the level of exploration, also known as the temperature parameter.

 Similar to \cite{Wang20}, we can show that the average of the sample trajectories $(Y^{\bm{\pi}}_t)_{t\geq 0}$ converges to $(\tilde{Y}_t^{\bm{\pi}})_{t\geq 0}$, which satisfies the following reflected SDE:    
\begin{align}\label{eq:tilde-Y-pi}
d \tilde{Y}_t^{\bm{\pi}}=\tilde{b}(\tilde{Y}^{\bm{\pi}}_t,\bm{\pi}(\cdot \mid \tilde{Y}^{\bm{\pi}}_t)) dt+\tilde{\sigma}(\tilde{Y}^{\bm{\pi}}_t,\bm{\pi}(\cdot \mid \tilde{Y}^{\bm{\pi}}_t)) dB_t+dL_t^{\bm{\pi}},
\end{align}
where $B=(B_t)_{t\geq 0}$ is a standard BM independent of $(\tilde{W}^{\kappa},W)$, the coefficient functions $(\tilde{b},\tilde{\sigma})$ are respectively defined by, for $(y,\bm{\pi})\in\R_+\times\mathcal{P}(\mathcal{A})$,
\begin{align*}
\begin{cases}
\displaystyle \tilde{b}(y,\bm{\pi}):=\int_{\mathcal{ A}} \tilde{\theta}^{\top}\mu\bm{\pi}(\tilde{\theta}\mid y) d\tilde{\theta}, \\[1.2em]
\displaystyle\tilde{\sigma}(y,\bm{\pi}):=\sqrt{\int_{ \mathcal{A}} \left(\frac{1}{2}\tilde{\theta}^{\top}\sigma\sigma^{\top}\tilde{\theta}-\sqrt{1-\kappa^2}\sigma_Z(y+1)\tilde{\theta}^{\top}\sigma\eta+\frac{\sigma_Z^2}{2}(y+1)^2\right)\bm{\pi}(\tilde{\theta}\mid y) d\tilde{\theta}}.
\end{cases}
\end{align*} 
Here, $L^{{\bm{\pi}}}=(L_t^{{\bm{\pi}}})_{t\geq0}$ denotes the local time process of the state process $\tilde{Y}^{\bm{\pi}}=(\tilde{Y}_t^{\bm{\pi}})_{t\geq0}$ at the level $0$,  namely, $L^{{\bm{\pi}}}=(L_t^{{\bm{\pi}}})_{t\geq0}$ is a continuous non-decreasing process satisfying $\int_{0}^t{\bf1}_{\tilde{Y}_s^{\bm{\pi}}=0}dL_s^{{\bm{\pi}}}=L_t^{{\bm{\pi}}}$ for all $t\geq0$. It follows from the property of Markovian projection in \cite{BS2013} that $Y_t^{\bm{\pi}}$ and $\tilde{Y}_t^{\bm{\pi}}$ have the same distribution for any $t\geq 0$. As a consequence, the value function \eqref{eq:obj-fuc} is equivalent to the following relaxed control form (see also \citealt{WKushner1990} and \citealt{Fleming1999}):
\begin{align}\label{eq:RL-val}
J(y ; \bm{\pi})=  \mathbb{E}^{\mathbb{Q}^W}\left[-\gamma\int_0^\infty e^{-\rho t} \int_{\mathcal{A}} \ln \bm{\pi}\left(\tilde{\theta} \mid  \tilde{Y}_t^{\bm{\pi}}\right) \bm{\pi}\left(\tilde{\theta}\mid \tilde{Y}_t^{\bm{\pi}}\right) d\tilde{\theta} dt- \int_0^\infty e^{-\rho t} dL_t^{\bm{\pi}}\Big| \tilde{Y}_0^{\bm{\pi}}=y\right].
\end{align}
The task of reinforcement learning is to find the optimal policy to attain the maximum of the value function that
\begin{align}\label{eq:optimal-RL-val}
v(y)=\max _{\bm{\pi}\in \Pi} J(y; \bm{\pi}),\quad y\geq0,
\end{align}
where the set $\Pi$ stands for the set of admissible (stochastic) policies. The following gives the precise definition of admissible policies.

\begin{definition}\label{def:admissible-pi}
A policy $\bm{\pi}$ is called admissible, that is, $\pi\in \Pi$, if
\begin{itemize}
\item[(i)]  $\bm{\pi}$ takes the feedback form as $\bm{\pi}_t={\bm \pi}(\cdot|Y_t)$ for $t\geq 0$, where $\bm{\pi}(\cdot|\cdot):{\cal A}\times \R_+\to \R$ is a measurable function and $\bm{\pi}(\cdot|y)\in{\cal P}( {\cal A})$ for all $y\in\R_+$;
\item[(ii)] the SDE \eqref{eq:tilde-Y-pi} admits a unique strong solution for any initial $y \in\R_+$.
\end{itemize}
\end{definition}

Assume that the value function \eqref{eq:RL-val} under a given admissible policy $\bm{\pi}\in\Pi$ is smooth enough, we can see that it satisfies the following Neumann problem, for $\bm{\pi}\in\mathcal{P}(\mathcal{A})$,
\begin{align}\label{RL}
\begin{cases}
\displaystyle \int_{\mathcal {A}}\left[H(\tilde{\theta},y,J_y(y;\bm{\pi}),J_{yy}(y;\bm{\pi}))
-\gamma \ln\bm{\pi}(\tilde{\theta}\mid y)\right]\bm{\pi}(\tilde{\theta}\mid y)d\tilde{\theta}\\[0.8em] \displaystyle\qquad\qquad\qquad\qquad\qquad\qquad\qquad+\frac{\sigma_Z^2 }{2}(y+1)^2J_{yy}(y;\bm{\pi})=\rho J(y;\bm{\pi}),\\[1em]
\displaystyle J_y(0;{\bm \pi})=1.
\end{cases}
\end{align}
On the other hand, the optimal value function defined by \eqref{eq:optimal-RL-val} satisfies the following exploratory HJB equation with a Neumann boundary condition that, for $y\geq0$,
\begin{align}\label{RL-HJB}
\begin{cases}
\displaystyle \sup_{{\bm \pi}\in{\mathcal{ P}}({\mathcal {A}})}\int_{\mathcal {A}}\left[H(\tilde{\theta},y,v_y(y),v_{yy}(y))
-\gamma \ln\bm{\pi}(\tilde{\theta})\right]\bm{\pi}(\tilde{\theta})d\tilde{\theta} +\frac{\sigma_Z^2}{2}(y+1)^2v_{yy}(y)=\rho  v(y),\\[1.7em]
\displaystyle v_y(0)=1.
\end{cases}
\end{align}
Here, the Hamilton operator $H:{\cal A}\times \R_+\times\R\times\R\to\R$ is defined by
\begin{align}\label{eq:Halmoper}
H(\tilde{\theta},y,P,Q):=\tilde{\theta}^{\top}\mu P+\frac{1}{2}\tilde{\theta}^{\top}\sigma\sigma^{\top}\tilde{\theta} Q-\sqrt{1-\kappa^2}\sigma_Z(y+1)\tilde{\theta}^{\top}\sigma\eta Q.
\end{align}
Moreover, by assuming $v_{yy}<0$, the candidate optimal policy is a Gaussian measure after normalization that 
\begin{align}\label{eq:optimal-candidate-policy}
{\bm{ \pi}}^*(\tilde{\theta}\mid y)&=\frac{ \exp\left\{\frac{1}{\gamma}H(\tilde{\theta},y,v_y(y),v_{yy}(y))\right\}}{ \int_{\mathcal {A}}\exp\left\{\frac{1}{\gamma}H(\tilde{\theta},y,v_y(y),v_{yy}(y))\right\}d\tilde{\theta} }.
%\nonumber\\
%&=\mathcal{N}\left(\tilde{\theta} ~\Big| -(\sigma\sigma^{\top})^{-1}\frac{\mu v_y(y)-\sqrt{1-\kappa^2}\sigma_Z(y+1)\sigma\eta v_{yy}(y)}{ v_{yy}(y)},-(\sigma\sigma^{\top})^{-1}\frac{\gamma}{v_{yy}(y)}\right),
\end{align}

In what follows, we first establish the policy improvement theorem. 

\begin{theorem}[Policy Improvement Theorem]\label{thm:PIT}
For any given $\bm{\pi} \in \Pi$, assume that the value function $J(\cdot;\bm{\pi})\in C^2(\R_+)$ satisfies Eq. \eqref{RL} with the Neumann boundary condition and $J_{yy}(\cdot;{\bm{\pi}})<0$ for any $y\in\R_+$. We consider the  mapping ${\cal I}$ on $\Pi$ defined by
\begin{align*}
\mathcal{I}(\bm{\pi})&:=\frac{ \exp\left\{\frac{1}{\gamma}H(\cdot,y,J_y(y;\bm{\pi}),J_{yy}(y;\bm{\pi}))\right\}}{ \int_{\mathcal {A}}\exp\left\{\frac{1}{\gamma}H(\tilde{\theta},y,J_y(y;\pi),J_{yy}(y;\bm{\pi}))\right\}d\tilde{\theta} }\nonumber\\
&=\mathcal{N}\left(\tilde{\theta} ~\Big| -(\sigma\sigma^{\top})^{-1}\frac{\mu J_y(y;\bm{\pi})-\sqrt{1-\kappa^2}\sigma_Z(y+1)\sigma\eta J_{yy}(y;\bm{\pi})}{ J_{yy}(y;\bm{\pi})},-(\sigma\sigma^{\top})^{-1}\frac{\gamma}{J_{yy}(y;\bm{\pi})}\right),
\end{align*}
where denote by ${\cal N}(\tilde{\theta}|a,b)$ the Gaussian density function with mean vector $a\in\R^d$ and covariance matrix $b\in\R^{d\times d}$. Then, we have
\begin{itemize}
\item[(i)]Denote by ${\bm \pi}^{\prime}={\cal I}(\bm{\pi})$.  If ${\bm\pi}^{\prime}\in \Pi$ satisfies
\begin{align}\label{eq:condition-pi}
\lim_{T\to\infty}\Ex^{\Qx^W}\left[e^{-\rho T}J(\tilde{Y}_T^{{\bm{\pi}}^{\prime}};\bm{\pi})\right]=0,
\end{align}
then it holds that $J\left(y;\bm{\pi}^{\prime}\right) \geq J(y;\bm{\pi})$ for all $y\in\R_+$.

\item[(ii)]If the map ${\cal I}$ has a fixed point $\bm{\pi}^*\in \Pi$ satisfying 
\begin{align}\label{eq:condition-optimal-pi}
\begin{cases}
\displaystyle \lim_{T\to\infty}\Ex^{\Qx^W}\left[e^{-\rho T}J(\tilde{Y}_T^{{\bm{\pi}}^{*}};\bm{\pi}^*)\right]=0,\\[1.2em]
\displaystyle \limsup_{T\to\infty}\Ex^{\Qx^W}\left[e^{-\rho T}J(\tilde{Y}_T^{{\bm{\pi}}};\bm{\pi}^*)\right]\geq 0,\quad \forall {\bm \pi}\in \Pi,
\end{cases}
\end{align}
then $\bm{\pi}^*$ is the optimal policy that $v(y)=\max _{\bm{\pi}\in \Pi} J(y; \bm{\pi})=J(y; \bm{\pi}^*)$.

\item [(iii)] In particular, if we choose the the temperature parameter $\gamma=\frac{\rho}{d}$, the map ${\cal I}$ has an explicit fixed point $\bm{\pi}^*\in\Pi$ given by
\begin{align}\label{eq:fixpoint-pi}
\bm{\pi}^* (\cdot|y)= \mathcal{N}\left(y~\Big|(\sigma\sigma^{\top})^{-1}(\mu+\sigma_z\sqrt{1-\kappa^2}\sigma \eta)(1+y),(\sigma\sigma^{\top})^{-1}\gamma(1+y)^2\right).
\end{align}
\end{itemize}
\end{theorem}

Theorem \ref{thm:PIT}-(i) provides a theoretical result for the policy improvement iteration;  while Theorem \ref{thm:PIT}-(ii) shows that the optimal policy is characterized by a fixed point of the iteration operator. In fact, Theorem \ref{thm:PIT}-(i) holds true for general Hamilton operator $H$ under some mild assumptions. However, it is usually difficult to verify the existence of fixed points and the transversality conditions \eqref{eq:condition-optimal-pi} for the general case. To address this issue, we choose a specific temperature parameter $\gamma=\rho/d$ in Theorem \ref{thm:PIT}-(iii), which significantly simplifies the finding of a fixed point of the mapping ${\cal I}$ by using the explicit expression \eqref{eq:fixpoint-pi}. The proof of Theorem \ref{thm:PIT}-(iii) also shows that if we start with a special Gaussian policy as $\bm{\pi}_0(\cdot|y)\sim \mathcal{N}(\cdot~|c_1(1+y),c_2(1+y)^2)$, then just after two steps of iterations, the  value function can no longer be improved, i.e., the fixed point is attained. We will mainly focus on this special choice of $\gamma=\rho/d$ and verify the transversality conditions \eqref{eq:condition-optimal-pi} in Section \ref{sec:example}, which implies that the fixed point \eqref{eq:fixpoint-pi} is the optimal policy.

Note that the policy improvement iteration in Theorem \ref{thm:PIT} depends on the knowledge of the model parameters, which are not known in the reinforcement learning procedure. Thus, in order to design an implementable algorithms, we turn to generalize the q-leaning theory initially proposed in \cite{JZ22c} for our purpose.

\subsection{q-function and martingale characterizations}

The aim of this section is to derive the q-function of our optimal tracking problem in continuous time and provide martingale characterizations of the q-function. We first give the definition of the q-function as the counterpart of the Q-function in the continuous time framework (c.f. \citealt{JZ22c}).
\begin{definition}[q-function]\label{def:q-function}
 The $q$-function of problem \eqref{eq:sample-Y}-\eqref{eq:obj-fuc} associated with a given policy $\bm{\pi} \in \Pi$ is defined as, for all $(y,\tilde{\theta}) \in\R_+ \times {\cal A}$,
\begin{align}\label{eq:q-function-def}
 q(y,\tilde{\theta};\bm{\pi}):=H\left(\tilde{\theta} ,y,J_y(y;\bm{\pi}), J_{yy}(y;\bm{\pi})\right)+\frac{\sigma_Z^2(y+1)^2}{2}J_{yy}(y;\bm{\pi})-\rho J(y;\bm{\pi}).
\end{align}
\end{definition}

 Moreover, we notice that the difference between the q function and the Hamiltonian is independent from $\theta$, and this allows one to use the q-function for a policy improvement theorem. Thus, the policy improvement mapping ${\cal I}$ in Theorem \ref{thm:PIT} can be expressed in terms of the q-function by $\mathcal{I}(\bm{\pi})=\frac{ \exp\left\{\frac{1}{\gamma}q(y,\cdot;\bm{\pi})\right\}}{ \int_{\mathcal {A}}\exp\left\{\frac{1}{\gamma}q(y,\tilde{\theta};\bm{\pi})\right\}d\tilde{\theta}}$.
This implies that the policy improvement iteration in Theorem \ref{thm:PIT} can be conducted by learning the q-function.

The following proposition gives the martingale condition to characterize the q-function for a given policy ${\bm \pi}$ when the value function is given.

\begin{proposition}\label{prop:martingale-condition}
Let a policy ${\bm \pi}\in \Pi$, its value function $J\in C^2(\R_+)$ satisfying Eq. \eqref{RL} and a continuous function $\hat{q}:\R_+ \times {\cal A} \rightarrow \R$ be given. Then, $\hat{q}(y,\tilde{\theta})=q(y,\tilde{\theta}; {\bm \pi})$ for all $(y,\tilde{\theta}) \in\R_+\times {\cal A}$ if and only if for all $y \in \R_+$, the following process
\begin{align}\label{eq:martinglae-J}
e^{-\rho t} J\left(Y_t^{\bm \pi}; {\bm \pi}\right)-\int_0^t e^{-\rho s}\hat{q}\left( Y_{s}^{\bm \pi},\tilde{\theta}_{s}^{\bm \pi}\right)d s-\int_0^t e^{-\rho s}dL_{s}^{\bm \pi},\quad t\geq0
\end{align}
is an $(\Fx,\mathbb{Q})$-martingale, where $Y^{\bm \pi}=(Y_t^{\bm \pi})_{t\geq0} $ is the solution to Eq.~\eqref{eq:sample-Y} with $Y_0^{\bm \pi}=y$. 
\end{proposition}

In the following, we strengthen Proposition \ref{prop:martingale-condition} and characterize the q-function and the value function associated with a given policy ${\bm \pi}$  simultaneously. This result is the crucial theoretical tool for designing the q-learning algorithm.
\begin{theorem}\label{thm:martingale-condition}
Let a policy ${\bm \pi} \in \Pi$, a function $\hat{J} \in C^2(\R_+)$ and a continuous function $\hat{q}:\mathbb{R}_+\times \mathcal{A} \to \mathbb{R}$ be given such that
\begin{align}
 &\qquad\quad\lim_{T\to\infty}\Ex^{\Qx}\left[e^{-\rho T}\hat{J}(Y_T^{\bm \pi})\right]=0, \label{eq:condition-J}\\
 &\int_{\mathcal{A}}[\hat{q}(y,\tilde{\theta})-\gamma \ln{\bm \pi}(\tilde{\theta}|y)] {\bm \pi} (\tilde{\theta}| y)d\tilde{\theta}=0, \quad \forall y\in\R_+.\label{eq:cindition-q}
\end{align}
Then, $\hat{J}$ and $\hat{q}$ are respectively the value function satisfying Eq. \eqref{RL} and the $q$-function associated with ${\bm \pi}$ if and only if for all $y\in \R_+$, the following process
\begin{align*}
e^{-\rho t} \hat{J}( Y_t^{\bm\pi})-\int_0^t e^{-\rho s}\hat{q}\left( Y_{s}^{\bm\pi},\tilde{\theta}_{s}^{\bm \pi}\right)d s-\int_0^t e^{-\rho s}dL_{s}^{\bm \pi},\quad t\geq0
\end{align*}
is an $(\Fx, \mathbb{Q})$-martingale, where $Y^{\bm \pi}=(Y_t^{\bm \pi})_{t\geq0}$ is the solution to Eq.~\eqref{eq:sample-Y} with $Y_0^{\bm \pi}=y$. If it holds further that ${\bm \pi}(\tilde{\theta}| y)=\frac{\exp\{\frac{1}{\gamma} \hat{q}(y,\tilde{\theta})\}}{\int_{\mathcal{A}} \exp\{\frac{1}{\gamma} \hat{q}(y,\tilde{\theta})\} d\tilde{\theta}}$ satisfying
\begin{align}\label{eq:condition-pi-geq}
 \limsup_{T\to\infty}\Ex^{\Qx}\left[e^{-\rho T}\hat{J}(Y_T^{{\bm{\pi}'}})\right]\geq 0,\quad \forall {\bm \pi}'\in \Pi,
\end{align}
then ${\bm \pi}$ is the optimal policy and $\hat{J}$ is the optimal value function.
\end{theorem}

\subsection{Continuous-time q-learning algorithm}
In this subsection, we design q-learning algorithms to simultaneously learn and update the parameterized value function and the policy based on the martingale condition in Theorem \ref{thm:martingale-condition}. 

Given a policy ${\bm \pi}\in \Pi$, we parameterize the value function by a family of functions $J^\xi(\cdot)$, where $\xi \in \Theta \subset \R^{L_\xi}$ and $L_{\xi}$ is the dimension of the parameter, and parameterize the q-function by a family of functions $q^\psi(\cdot,\cdot)$, where $\psi \in \Psi \subset \R^{L_\psi}$ and $L_{\psi}$ is the dimension of the parameter.  Moreover,  we shall make the following assumptions for the parameterized family $\{J^{\xi}\}$ and $\{q^{\psi}\}$.
\begin{itemize}
\item[{\bf(A$_{\xi}$)}]  The family of functions $J^\xi(\cdot)$  is $C^1$ in $\xi$ and satisfies $\lim_{t\to \infty}\Ex^{\Qx}[e^{-\rho t} J^{\xi}(Y_t^{\bm \pi})]=0$ and the Neumann condition $J_y^{\xi}(0)=1$. Moreover, there exist continuous functions $G_J(\cdot)$ and $\tilde{J}(\cdot)$ such that $|J^{\xi}(y)|+|J_y^{\xi}(y)|+|J_{yy}^{\xi}(y)|\leq G_J(\xi)\tilde{J}(y)$ for all $(\xi,y)\in \Theta\times\R_+$.

\item[{\bf(A$_{\psi}$)}]   The family of functions $q^\psi(\cdot,\cdot)$ is $C^1$ in $\psi$ and satisfies that
\begin{align}
\int_{\mathcal{A}}[q^{\psi}(y,\tilde{\theta})-\gamma \ln{\bm \pi}(\tilde{\theta}\mid y)] {\bm \pi} (\tilde{\theta}\mid y)d\tilde{\theta}=0, \quad \forall y\in\R_+.\label{eq:cindition-q-xi}
\end{align}
 Furthermore, there exist continuous functions $G_q(\cdot)$ and $\tilde{q}(\cdot,\cdot)$ such that
\begin{align}\label{eq:assumption-q}
|q^{\psi}(y,\tilde{\theta})|\leq G_q(\psi)\tilde{q}(y,\tilde{\theta}),\quad\forall(\psi,y,\tilde{\theta})\in \Psi\times\R_+\times {\cal A}.
\end{align}
\end{itemize}
Then, the learning task is to find the ``optimal'' (in some sense) parameters $\xi$ and $\psi$.  The key step in the algorithm design is to enforce the martingale condition stipulated in Theorem \ref{thm:martingale-condition}. More precisely, let $M=(M_t)_{t\geq 0}$ be the  martingale given in  Theorem \ref{thm:martingale-condition}, i.e.,
\begin{align}
M_t&=e^{-\rho t} J\left(Y_t^{\bm \pi}\right)-\int_0^t e^{-\rho s}q\left( Y_{s}^{\bm \pi},\tilde{\theta}_{s}^{\bm \pi}\right)d s-\int_0^t e^{-\rho s}dL_{s}^{\bm \pi}\label{eq:M},
\end{align}
and $M^{\xi,\psi}=(M^{\xi,\psi}_t)_{t\geq 0}$ be its parameterized process defined by
\begin{align}
M_t^{\xi,\psi}&=e^{-\rho t} J^{\xi}\left(Y_t^{\bm \pi}\right)-\int_0^t e^{-\rho s}q^{\psi}\left( Y_{s}^{\bm \pi},\tilde{\theta}_{s}^{\bm \pi}\right)d s-\int_0^t e^{-\rho s}dL_{s}^{\bm \pi}\label{eq:M-xi-psi},
\end{align}
where $Y^{\bm \pi}=(Y_t^{\bm \pi})_{t\geq0} $ is the solution to \eqref{eq:sample-Y} with $Y_0^{\bm \pi}=y$. 

It follows from the martingale orthogonality condition that, for any test  adapted continuous process $\varsigma=(\varsigma_t)_{t\geq0}$ with $\Ex^{\Qx}[\int_0^{\infty}|\varsigma_t|^2d\left<M\right>_t]<\infty$, 
\begin{align}\label{eq:orthogonality-M}
\Ex^{\Qx}\left[\int_0^{\infty}\varsigma_tdM_t\right]=0.
\end{align}
In fact, the following result shows that this is a necessary and sufficient condition for the parameterized process $M^{\xi,\psi}$ to be a martingale. Its proof is omitted because it is similar to the one of Proposition 4 in \cite{JZ22b}.
\begin{proposition}\label{prop:orthogonality}
The parameterized process $M^{\xi,\psi}$ given by \eqref{eq:M-xi-psi} is a martingale if and only if
\begin{align}\label{eq:orthogonality}
\Ex^{\Qx}\left[\int_0^{\infty}\varsigma_tdM_t^{\xi,\psi}\right]=0,
\end{align}
 for any $\varsigma$ with $\Ex^{\Qx}[\int_0^{\infty}|\varsigma_t|^2d\left<M\right>_t]<\infty$.
\end{proposition}

Proposition \ref{prop:orthogonality} tells that, to find the ``optimal'' parameters $\xi$ and $\psi$, it is enough to explore the solution $(\xi^*,\psi^*)$ of the martingale orthogonality equation \eqref{eq:orthogonality}. This can be implemented by using stochastic approximation to update parameters as, for $\chi\in\{\xi,\psi\}$,
\begin{align}\label{eq:SA}
 \chi \leftarrow \chi+\alpha_\chi \int_0^{\infty} \varsigma_tdM_t^{\xi,\psi},
\end{align}
where $\alpha_\chi>0$ is the learning rate. However,  the martingale orthogonality equation \eqref{eq:orthogonality} and the associated update rule \eqref{eq:SA} may not be directly applicable due to its requirement of having full trajectories of the infinite horizon. To overcome this difficulty, we first truncate the martingale orthogonality equation at a sufficiently long time $T$:
\begin{align}\label{eq:orthogonality-finite}
\Ex^{\Qx}\left[\int_0^{T}\varsigma_tdM_t^{\xi,\psi}\right]=0.
\end{align}
Note that the truncated martingale orthogonality equation \eqref{eq:orthogonality-finite} involves integration in continuous time, which still cannot be directly applicable. Thus, we turn to consider truncated discrete-time martingale orthogonality equation.  Let $K\in\mathbb{N}$ and $\Delta t=T/K$, consider the partition $0=t_0<t_1<t_2<\cdots<t_K=T$ with $t_k-t_{k-1}=\Delta t$ for $k=1,...,K$. In light of \eqref{eq:orthogonality-finite}, we present the truncated discretized martingale orthogonality equation as
\begin{align}\label{eq:orthogonality-finite-discrete}
\Ex^{\Qx}\left[\sum_{k=0}^{K-1}\varsigma_{t_k}\left(M_{t_{k+1}}^{\xi,\psi}-M_{t_k}^{\xi,\psi}\right)\right]=0.
\end{align}

The next theorem states the convergence  of the stochastic approximation algorithms when $\Delta t\to 0$ and $T\to \infty$.
\begin{theorem}\label{thm:convergence-policy}
 For a continuous test function $\varsigma$ satisfying $\Ex^{\Qx}[\int_0^{\infty}|\varsigma_t|^2d\left<M\right>_t]<\infty$, consider the martingale orthogonality condition equations \eqref{eq:orthogonality}, \eqref{eq:orthogonality-finite} and \eqref{eq:orthogonality-finite-discrete}.  Then, we have
\begin{itemize}
\item[(i)]  any convergent subsequence of the solution  $(\xi_{\mathrm{ML}}^*(T),\psi_{\mathrm{ML}}^*(T))$ that solves the truncated martingale orthogonality equation \eqref{eq:orthogonality-finite} converges to the solution of the martingale  orthogonality equation \eqref{eq:orthogonality}, that is
\begin{align*}
\lim_{T\to \infty}(\xi_{\mathrm{ML}}^*(T),\psi_{\mathrm{ML}}^*(T))=(\xi_{\mathrm{ML}}^*,\psi_{\mathrm{ML}}^*)
\end{align*}
solves Eq.~\eqref{eq:orthogonality}.
\item[(ii)] any convergent subsequence of the solution $(\xi_{\mathrm{ML}}^*(\Delta t;T),\psi_{\mathrm{ML}}^*(\Delta t;T))$ that solves the truncated discretized  martingale  orthogonality equation  \eqref{eq:orthogonality-finite-discrete} converges to the solution of the truncated martingale orthogonality equation \eqref{eq:orthogonality-finite}, that is
\begin{align*}
\lim _{\Delta t \rightarrow 0} (\xi_{\mathrm{ML}}^*(\Delta t;T),\psi_{\mathrm{ML}}^*(\Delta t;T))=(\xi_{\mathrm{ML}}^*(T),\psi_{\mathrm{ML}}^*(T)) 
\end{align*}
solves Eq.~\eqref{eq:orthogonality-finite}.
\end{itemize}
\end{theorem}
A direct result is that  if the solution  $(\xi_{\mathrm{ML}}^*(\Delta t;T),\psi_{\mathrm{ML}}^*(\Delta t;T))$ of the truncated discretized martingale orthogonality equation  \eqref{eq:orthogonality-finite-discrete}  converges, then it holds that
\begin{align*}
\lim_{T\to \infty}\lim _{\Delta t \rightarrow 0} (\xi_{\mathrm{ML}}^*(\Delta t;T),\psi_{\mathrm{ML}}^*(\Delta t;T))=(\xi_{\mathrm{ML}}^*,\psi_{\mathrm{ML}}^*) 
\end{align*}
solves the martingale  orthogonality equation \eqref{eq:orthogonality}. Therefore, it provides a theoretical foundation for implementing the truncation and discretization in the learning algorithm.

Next we use the truncated discretized martingale orthogonality equation \eqref{eq:orthogonality-finite-discrete}, Theorem \ref{thm:PIT}, Theorem \ref{thm:martingale-condition}, Theorem \ref{thm:convergence-policy} to design the q-learning algorithm. To learn the optimal value function and q-function, we choose the proper parameterized function approximators $J^{\xi}$ and $q^{\psi}$ for $\xi \in \Theta \subset \R^{L_\xi}, \psi \in \Psi \subset \R^{L_\psi}$, which satisfy the assumptions {\bf(A$_{\xi}$)} and {\bf(A$_{\psi}$)}. Theorem \ref{thm:PIT} and the definition of q-function the parameterized form of the optimal policy as
\begin{align}\label{eq:pis-pi}
{\bm \pi}^{\psi}(\tilde{\theta}\mid y)=\frac{\exp \left\{\frac{1}{\gamma} q^{\psi}(y,\tilde{\theta} )\right\}}{\int_{\cal A} \exp \left\{\frac{1}{\gamma} q^{\psi}(y,\tilde{\theta})\right\}d\tilde{\theta} },\quad \forall (y,\tilde{\theta})\in \R_+\times {\cal A}.
\end{align}
Moreover, it can be easily see that with the policy ${\bm \pi}^{\psi}$ given in \eqref{eq:pis-pi}, the consistency condition in the assumption {\bf(A$_{\psi}$)} trivially holds that 
\begin{align}\label{eq:constraint}
 \int_{\cal A}\left[q^{\bm{\psi}}(y,\tilde{\theta})-\gamma \ln {\bm \pi}^{\psi}(\tilde{\theta}\mid y)\right]{\bm  \pi}^{\psi}(\tilde{\theta}\mid y) d\tilde{\theta}=0.
\end{align}
Based on the truncated discretized martingale orthogonality equation \eqref{eq:orthogonality-finite-discrete}, we can apply the stochastic approximation method \cite{Robbins1951} to update the parameters in the following manner, which corresponds to the TD(0) algorithm (c.f. \cite{Sutton1988} and \cite{JZ22b}):
\begin{align}\label{eq:SGD-continuous}
\xi \leftarrow \xi+\alpha_\xi \sum_{k=0}^{K-1}\iota_{t_k} G_k,\quad \psi \leftarrow \psi+\alpha_\psi \sum_{k=0}^{K-1}\varsigma_{t_k} G_k,
\end{align}
where $\alpha_{\xi},\alpha_{\psi}>0$ are the learning rates, $\iota,\varsigma$ are test functions, and for $k=0,1,\ldots,K-1$, the quantity $G_k$ is given by
\begin{align*}
G_k:=e^{-\rho t_{k+1}}J^\xi\left(Y_{t_{k+1}}^{\pi^{\bm \psi}}\right)-e^{-\rho t_{k}}J^\xi\left(Y_{t_k}^{\pi^{\bm \psi}}\right)-e^{-\rho t_{k}}q^{\bm{\psi}}\left(Y_{t_k}^{\pi^{\bm{\psi}}},\theta_{t_k}^{\pi^{\bm \psi}}\right) \Delta t-e^{-\rho t_{k}}\left(L_{t_{k+1}}^{\pi^{\bm{\psi}}}-L_{t_{k}}^{\pi^{\bm{\psi}}}\right) .
\end{align*}
Note that, by using the martingale orthogonality condition, we need at least $L_{\xi}+L_{\psi}$ equations to fully determine the parameters $(\xi,\psi)$. In this paper, we choose the test functions in the conventional sense that
\begin{align*}
\iota_t=e^{-\rho t}\frac{\partial J^\xi}{\partial \xi}\left(Y_t^{\pi^{\bm \psi}}\right),\quad \varsigma_t= e^{-\rho t}\frac{\partial q^\psi}{\partial \psi}\left( Y_s^{\bm{\pi}^\psi}, \theta_s^{\boldsymbol{\pi}^\psi}\right)  .
\end{align*}

Based on the above updating rules, we then present the pseudo-code of the offline q-learning algorithm in Algorithm \ref{Alg:tracking}. 
\ \\
\ \\

	\begin{algorithm}[h]
			\caption{\textbf{Offline q-Learning Algorithm}}
		    \label{Alg:tracking}
			\hspace*{0.02in} {\bf Input:} 
			 Initial state $y_0$, horizon $T$, time step $\Delta t$, number of episodes $N$, number of mesh grids $K$, initial learning rates $\alpha_\xi(\cdot), \alpha_\psi(\cdot)$  (a function of the number of episodes), functional forms of parameterized value function $J^\xi(\cdot)$, q-function $q^\psi(\cdot)$, policy $\boldsymbol{\pi}^\psi(\cdot \mid \cdot)$ and temperature parameter $\gamma$.\\
			\hspace*{0.02in} {\bf Required Program:} an environment simulator $(y^{\prime}, L^{\prime})=$ Environment $_{\Delta t}(t,y,L,\tilde{\theta})$ that takes current time-state pair $(t,y,L)$ and action $\theta$ as inputs and generates state $y^{\prime}$ and $L^{\prime}$ at time $t+\Delta t$ as outputs . \\
			\hspace*{0.02in} {\bf Learning Procedure:}
			\begin{algorithmic}[1]
\State Initialize $\xi$,~$\psi$ and $i=1$. 
				\While{$i<N$}  
				\State Initialize $j = 0$.  Observe initial state $y_0$ and store $y_{t_0}\leftarrow y_0$.
				\While{$j < K$}
				  
				    \State \ \ \ \ \    Generate action $\tilde{\theta}_{t_j} \sim \bm{\pi}^\psi\left(\cdot \mid y_{t_j}\right)$.
                   \State \ \ \ \ \   Apply $\tilde{\theta}_{t_j} $ to environment simulator $(y, L)=$ Environment  $_{\Delta t}(t_j, y_{t_j},L_{t_j},\tilde{\theta}_{t_j})$.
                   \State \ \ \ \ \   Observe  new state $y$ and $L$ as output. Store $y_{t_{j+1}} \leftarrow y$,  and $ L_{t_{j+1}} \leftarrow  L$. 
                \EndWhile{} 
                \State  For every $k=0,1,...,K-1$, compute 
\begin{align*}
 G_{k}&=J^\xi\left(y_{t_{k+1}}\right)-J^\xi\left(y_{t_k}\right)-q^{\bm{\psi}}\left(y_{t_k},\tilde{\theta}_{t_k}\right) \Delta t-\left(L_{t_{k+1}}-L_{t_{k}}\right)-\rho J^\xi\left(y_{t_k}\right)\Delta t.
\end{align*}

				\State Update $\xi$ and $\psi$ by
\begin{align*}
\xi &\leftarrow \xi+\alpha_\xi(i) \sum_{k=0}^{K-1} e^{-\rho t_k} \frac{\partial J^\xi}{\partial \xi}\left(y_{t_k}\right)G_{k}, \\
\psi &\leftarrow \psi+\alpha_\psi(i) \sum_{k=0}^{K-1}  e^{-\rho t_k} \frac{\partial q^\psi}{\partial \psi}\left(y_{t_k}, \tilde{\theta}_{t_k}\right)G_{k}.
\end{align*}
               \State   Update $i \leftarrow i+1$.
				          
              \EndWhile{}                
    \end{algorithmic}
\end{algorithm}

\section{Numerical Examples}\label{sec:example}

In this section, we consider our auxiliary control problem with reflections as an example beyond the LQ control framework to illustrate the proposed q-learning algorithm. In particular, we will first choose the temperature parameter $\gamma=\rho/d$, where we recall that $\rho$ is the discount factor and $d$ is number of risky assets. We also note that, in this case, Theorem \ref{thm:PIT}-(iii) asserts that the map ${\cal I}$ has a fixed point $\bm{\pi}^*$ given by \eqref{eq:fixpoint-pi}. However, it remains to verify the transversality condition \eqref{eq:condition-optimal-pi} to check the optimality of $\bm{\pi}^*$. 

In what follows, we first give the explicit solution of the exploratory HJB equation \eqref{eq:exploratory-HJB} and shows that $\bm{\pi}^*$ is indeed the optimal policy. 
Plugging \eqref{eq:optimal-candidate-policy} into the exploratory HJB equation \eqref{RL-HJB}, we get that
\begin{align}\label{eq:exploratory-HJB}
 &\frac{1}{2}\sigma_Z^2\kappa^2 (y+1)^2 v_{yy}(y)+\sqrt{1-\kappa^2}\zeta(y+1)v_y(y)+\gamma \ln\left(\sqrt{-\frac{(2\pi \gamma)^d}{|\sigma\sigma^{\top}| (v_{yy}(y))^d}}\right)\nonumber\\
&\qquad-\alpha\frac{(v_y(y))^2}{v_{yy}(y)}=\rho v(y)
\end{align}
 with the Neumann boundary condition $v_y(0)=1$. We also recall that  $\alpha=\frac{1}{2}\mu^{\top}(\sigma\sigma^{\top})^{-1}\mu$ and $\zeta=\sigma_Z \eta^{\top}\sigma^{-1}\mu$.
We conjecture that $v(y)$ for $y\geq0$ satisfies the form of
\begin{align}\label{eq:JAB}
v(y)=A\ln(1+y)+B,
\end{align}
for some constants $A$ and $B$. By plugging the expression into \eqref{eq:exploratory-HJB}, \eqref{eq:JAB} and using the choice of $\gamma=\rho/d$, we obtain that
\begin{align*}
A=1, \quad B=\frac{1}{\rho}\left(-\frac{1}{2}\sigma_Z^2\kappa^2+\sqrt{1-\kappa^2}\zeta+\frac{\gamma}{2}\ln\left(\frac{(2\pi \gamma)^d}{|\sigma\sigma^{\top}|}\right)+\alpha\right).
\end{align*}
It then follows that
\begin{align}\label{eq:optimal-J}
v(y)=\ln(1+y)+\frac{1}{\rho}\left(-\frac{1}{2}\sigma_Z^2\kappa^2+\sqrt{1-\kappa^2}\zeta+\frac{\gamma}{2}\ln\left(\frac{(2\pi \gamma)^d}{|\sigma\sigma^{\top}|}\right)+\alpha\right)
\end{align}
is a classical solution of Eq.~\eqref{eq:exploratory-HJB}. For the general temperature parameter $\gamma\neq \rho/d$, the classical solution to the exploratory HJB equation ~\eqref{eq:exploratory-HJB}, if it exists, is not in the simple form of \eqref{eq:optimal-J} and actually does not admit an explicit expression, which shows the significant influence of the entropy regularizer in our stochastic control problem when it is not the LQ type control.

Back to the special choice of $\gamma=\rho/d$, as a result of \eqref{eq:optimal-candidate-policy} and \eqref{eq:optimal-J}, the candidate optimal policy is given by the following Gaussian policy 
\begin{align}\label{eq:optimal-pi}
\bm{\pi}^* (\cdot|y)= \mathcal{N}\left(\tilde{\theta}\mid(\sigma\sigma^{\top})^{-1}(\mu+\sigma_z\sqrt{1-\kappa^2}\sigma \eta)(1+y),(\sigma\sigma^{\top})^{-1}\gamma(1+y)^2\right),
\end{align}
which is also the fixed point as shown in Theorem \ref{thm:PIT}-(iii). It is notable that both the mean and the variance of the Gaussian policy $\bm{\pi}^* (\cdot|y)$ depend on the state variable $y$, and in particular, its variance is increasing in $y$. That is, for the larger state process $Y_t$ or the state process $X_t$, the agent needs to implement the more random policies with the larger variance for the purpose of learning.

The following verification theorem shows that the classical solution of the exploratory HJB equation \eqref{RL-HJB} coincides with the value function and provides the optimal policy.
\begin{theorem}[Verification Theorem]\label{thm:verification}
Consider  the classical solution $v(y)$ for $y\in\R_+$ of the exploratory HJB equation \eqref{eq:exploratory-HJB} given by \eqref{eq:optimal-J}, the policy $\bm{\pi}^*$ given by \eqref{eq:optimal-pi}, and the controlled state process $Y^*=(Y_t^*)_{t \geq 0}$ that obeys the following reflected SDE, for all $t \geq 0$,
\begin{align}\label{eq:optimal-Y-pi}
d Y_t^{*}&=(2\alpha+\sqrt{1-\kappa^2}\zeta) (1+Y^*_t) dt+\sqrt{\alpha+\frac{d}{2}\gamma+\frac{1}{2}\kappa^2\sigma_Z^2+\sqrt{1-\kappa^2}\zeta}(1+Y^*_t)dB_t+dL_t^{*}.
\end{align}
Here, $L^*=(L_t^*)_{t\geq 0}$ is the local time process for the process $Y^*=(Y_t^*)_{t\geq 0}$ at the reflecting boundary $0$, namely, $L^*=(L_t^*)_{t\geq 0}$ is a continuous non-decreasing process satisfying $\int_{0}^t{\bf1}_{Y_s^{*}=0}dL_s^{*}=L_t^{*}$. Then, $\bm{\pi}^*$ is an optimal policy and $v(y)$ is the value function. That is, for all admissible ${\bm \pi} \in \Pi$, we have $J(y;{\bm\pi}) \leq v(y)$ for all $y\in\R_+$, where the equality holds when ${\bm \pi}={\bm \pi}^*$.
%\begin{align*}
%J(y;{\bm\pi}) \leq v(y), \quad \forall y\in\R_+
%\end{align*}

\end{theorem}

\begin{figure}[h]
\centering
\includegraphics[width=7cm]{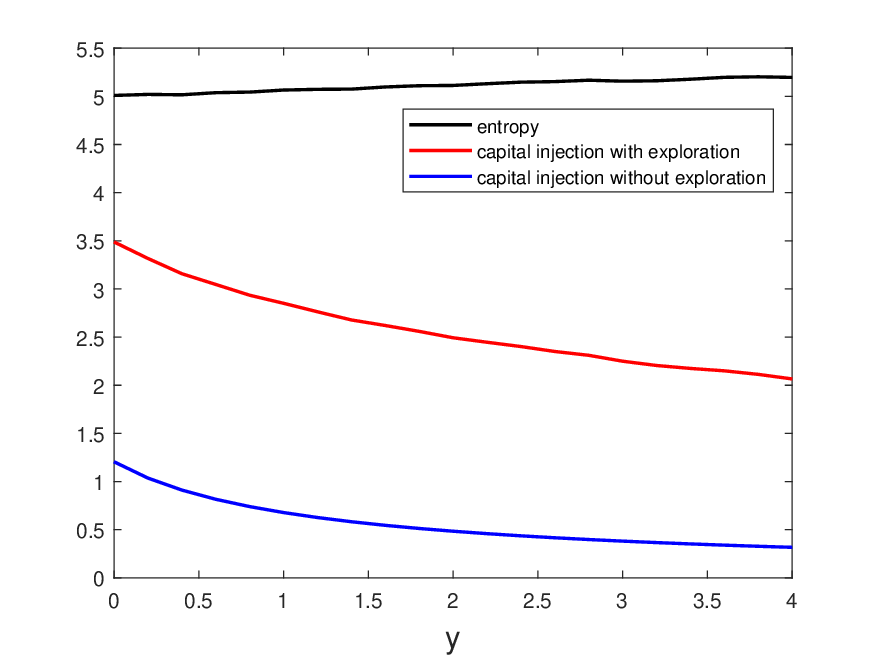}
\caption{ \small {Expected entropy and expected capital injection.} 
%The parameters are set to be $d=1,~\mu=0.2,~\sigma=0.5,\sigma_Z=2,~\rho=3,~\kappa=1$.
}\label{fig:entropy}
\end{figure}

\begin{remark}
When the model parameters are known, we present in 
Figure \ref{fig:entropy} the numerical comparison between the expectation of the discounted capital injection with exploration and the expectation of the discounted capital injection without exploration, illustrating the influence by the additional entropy regularizer or the policy exploration. As shown in Figure \ref{fig:entropy},  both expectations of the discounted capital injection with exploration and without exploration are decreasing with respect to the variable $y$, which can be explained by the fact that the larger value of $y$ indicates the larger initial distance between the portfolio process $V^\theta$ and the benchmark process $Z$, thereby it will have smaller chances for the capital injection. More importantly, Figure \ref{fig:entropy} shows that if the fund manager employs the exploratory policy to
learn the optimal policy in the unknown environment, the randomized actions will inevitably cause the underlying controlled processes to hit the reflection boundary $0$ more often comparing with the case using the strict control, thereby leading to the higher expected capital injection.
\end{remark}

Let us denote some parameters $(\xi^*,\psi_1^*,\psi_2^*,\psi_3^*)$  by
\begin{align}\label{eq:optimal-psi-xi}
\begin{cases}
\displaystyle \xi^*=\frac{1}{\rho}\left(-\frac{1}{2}\sigma_Z^2\kappa^2+\sqrt{1-\kappa^2}\zeta+\frac{\gamma}{2}\ln\left(\frac{(2\pi \gamma)^d}{|\sigma\sigma^{\top}|}\right)+\alpha\right),\\[0.9em]
\displaystyle \psi^*_1=\mu+\sigma_z\sqrt{1-\kappa^2}\sigma \eta,~~~\psi_2^*=\sigma,\\[0.5em]
\displaystyle\psi_3^*=\frac{1}{2}\sigma_Z^2(\kappa^2-1)-\alpha-\sqrt{1-\kappa^2}\zeta-\frac{\gamma}{2}\ln\left(\frac{(2\pi \gamma)^d}{|\sigma\sigma^{\top}|}\right).
\end{cases}
\end{align}

Using Definition \ref{def:q-function}, the q-function can be expressed by, for all $(y,\tilde{\theta})\in\R_+\times{\cal A}$,
\begin{align}\label{eq:q-function}
q(y,\tilde{\theta})&=\frac{(\psi_1^*)^{\top}\tilde{\theta}}{1+y}-\frac{\tilde{\theta}^{\top}\psi_2^*(\psi_2^*)^{\top}\tilde{\theta}}{2(1+y)^2}-\rho \ln(1+y)+\psi_3^*,
\end{align}
and the value function $v(y)$ can be written by
\begin{align}\label{vxi}
v(y)=\ln(1+y)+\xi^*,\quad\forall y\in\R_+.
\end{align}

Based on \eqref{eq:q-function}-\eqref{vxi}, for all $(y,\tilde{\theta})\in\R_+\times {\cal A}$, we can parameterize the optimal value function and the optimal q-function in the exact form as: 
\begin{align}\label{eq:parameter-J-q}
\begin{cases}
\displaystyle J^{\xi}(y)=\ln(1+y)+\xi, \\[0.6em]
\displaystyle q^{\bm \psi}(y,\tilde{\theta})=\frac{ \psi_1^{\top}\tilde{\theta}}{1+y}- \frac{\tilde{\theta}^{\top} \psi_2\psi_2^{\top}\tilde{\theta}^{\top}}{2(1+y)^2}-\rho \ln(1+y)+\psi_3
\end{cases}
\end{align}
with the parameters $(\xi,\psi_1,\psi_2)\in\R\times \R^d \times\R^{d\times d}$ to be learnt, the parameter $\psi_3$ satisfying 
\begin{align*}
\psi_3=-\frac{1}{2}\psi_1^{\top}(\psi_2\psi_2^{\top})^{-1}\psi_1-\frac{\gamma}{2}\ln\left(\frac{(2\pi \gamma)^d}{|\psi_2\psi_2^{\top}|}\right),
\end{align*}
 and the optimal policy can be parameterized by
$\pi^{\bm{\psi}}(\theta\mid y)=\frac{\exp \left\{\frac{1}{\gamma} q^{\bm{\psi}}(y,\theta )\right\}}{\int_{\cal A} \exp \left\{\frac{1}{\gamma} q^{\bm{\psi}}(y,\theta)\right\}d\theta}$. We can verify that the parameterized value function and q-function satisfy assumptions {\bf(A$_{\xi}$)} and {\bf(A$_{\psi}$)}. 

We consider the model with one risky asset (i.e., $d = 1$), and we set the coefficients of the simulator to be $\mu=0.2$, $\sigma=1$, $\sigma_Z=0.2$ and $\kappa=0.5$. Furthermore, we set $\gamma=\rho=0.2$ and the truncated horizon $T=12$, the time step $\Delta t=0.005$ and the number of episodes $N=2\times 10^4$. The learning rates $(\alpha_\psi,\alpha_\xi)$ are chosen by
\begin{align*}
\alpha_{\psi}(i)&=
\begin{cases}
\displaystyle \left(\frac{0.1}{i^{0.61}},\frac{0.01}{i^{0.61}}\right), &1\leq i\leq 1\cdot 10^4,\\[0.7em]
\displaystyle \left(\frac{0.05}{i^{0.81}},\frac{0.005}{i^{0.81}}\right) , &1\cdot 10^4< i\leq 2\cdot 10^4,
\end{cases}\quad
\alpha_{\xi}(i)&=
\begin{cases}
\displaystyle \frac{0.015}{i^{0.61}}, &1\leq i\leq 1\cdot 10^4,\\[0.7em]
\displaystyle \frac{0.005}{i^{0.81}}, &1\cdot 10^4< i\leq 2\cdot 10^4,
\end{cases}
\end{align*}
which decay along the iterations. Based on Algorithm \ref{Alg:tracking}, we plot in Figure \ref{fig:convergence} the numerical results on the
convergence of iterations for parameters $(\psi_1,\psi_2,\psi_3,\xi)$ in the optimal value function and the optimal q-function. The learnt parameters for the optimal value function and the optimal
 q-function are summarized in Table \ref{table:learnt-parameter}.

\begin{table}[h]   
\begin{center}   
\caption{The true value and learnt value of the parameters $(\psi_1,\psi_2,\psi_3,\xi)$.}  
\label{table:learnt-parameter} 
\begin{tabular}{ m{3cm}<{\centering} |m{2.2cm}<{\centering}| m{2.2cm}<{\centering}| m{2.2cm}<{\centering}| m{2.2cm}<{\centering}}   
\hline & $\psi_1$ & $\psi_2$ & $\psi_3$  &$\xi$ \\   
\hline  True value &0.3732 &1.0000 &-0.0050  &0.3624 \\ 
\hline   Learnt value &0.2520 &1.2410 &-0.0002  &0.3326\\  
\hline
\end{tabular}   
\end{center}   
\end{table}

\begin{figure}[h]
\centering
  \subfigure[]{
        \includegraphics[width=5cm]{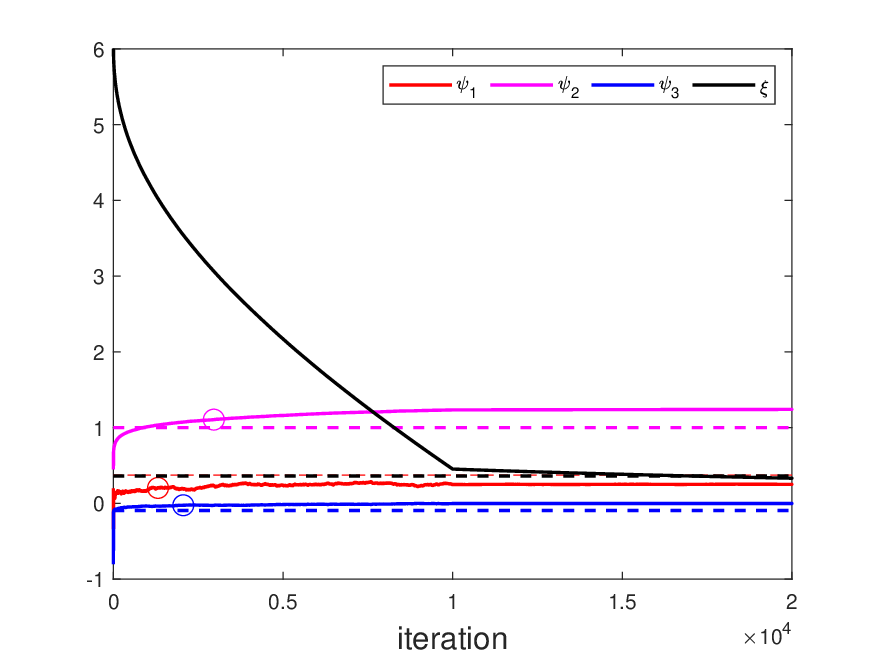}
    }\hspace{-8mm}
  \subfigure[]{
        \includegraphics[width=5cm]{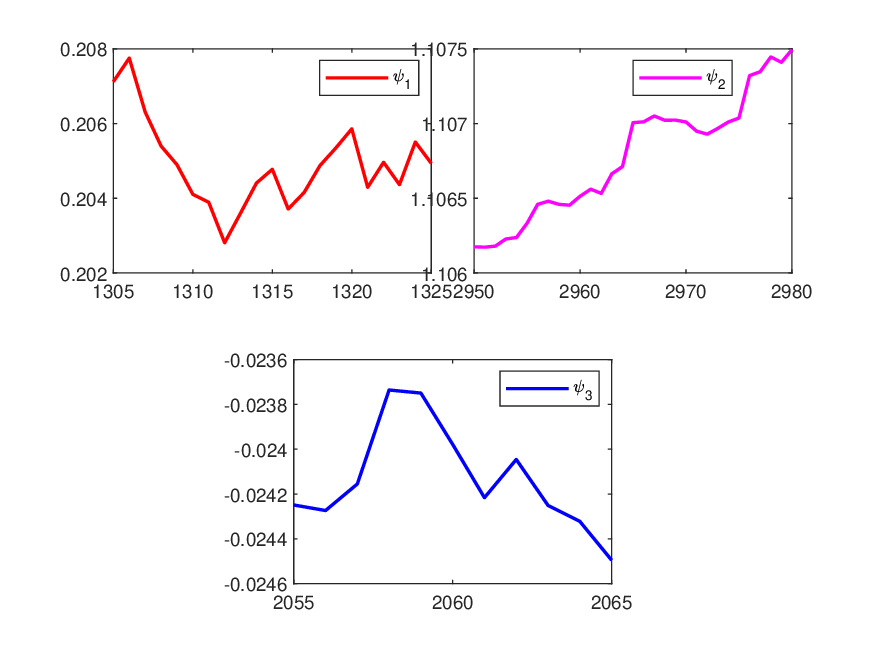}
    }\hspace{-8mm}
\subfigure[]{
        \includegraphics[width=5cm]{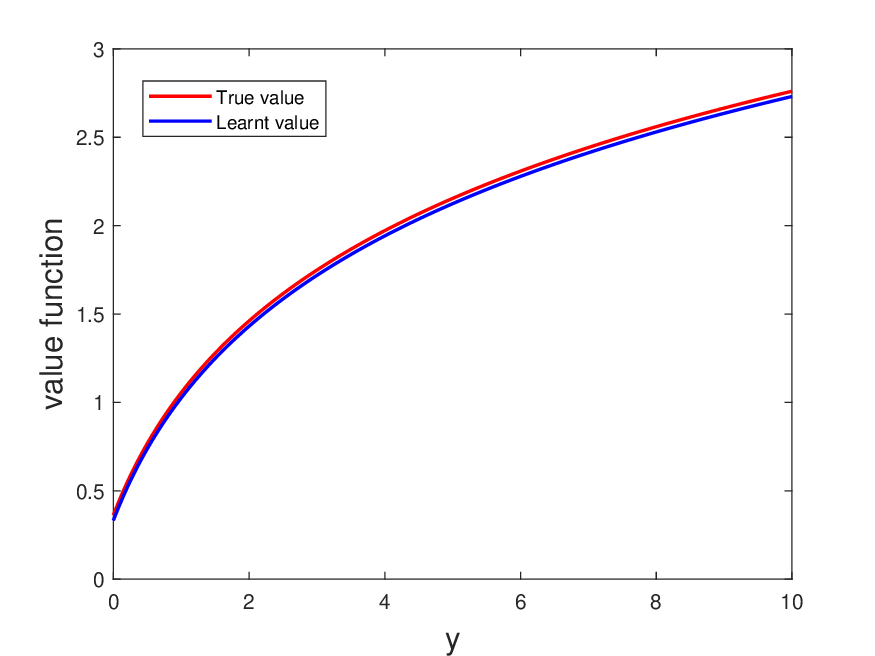}
    }
\caption{ \small {Convergence of parameter iterations in the optimal value function and the optimal q-function using Algorithm \ref{Alg:tracking}.}
 (a): paths of learnt parameters for the optimal value function and the optimal q-function vs optimal parameters shown in the dashed line; (b): ``zoom-in" paths of learnt parameters for the hollow circles selected in panel (a) to illustrate the fluctuations of the local path; (c): the learnt value function vs the optimal value function.}\label{fig:convergence}
\end{figure}

For the general choice of temperature parameter $\gamma\neq \rho/d$, one needs to apply neural networks as the parameterized approximations of the optimal value function and the optimal q-function. In what follows, we implement our proposed q-learning algorithm together with neural networks. More precisely, we consider the parameterized value function $J^{\xi}(y)$ with $J^{\xi}_y(0)=1$ and parameterized q-function given by $q^{\psi}(y,\tilde{\theta})=-a^{\psi}(y)(\tilde{\theta}-b^{\psi}(y))^2+c^{\psi}(y)$ with  $c^{\psi}(y)=-\frac{\gamma}{2}\ln\big(\frac{(2\pi\gamma)^d}{|b^{\psi}(y)|}\big)$ 
 for $(y,\tilde{\theta})\in\R_+\times {\cal A}$.  Thus, we adopt three two-layer fully connected neural networks with Sigmoid activation functions to train the functions $J^{\xi}(y)$ and $(a^{\psi}(y),b^{\psi}(y))$. Consider the model with one risky asset (i.e., $d = 1$), and still set the coefficients of the simulator to be $\mu=0.2$, $\sigma=1$, $\sigma_Z=0.2$, $\kappa=0.5$ and  $\rho=0.2$. Here, we choose different temperature parameters $\gamma=0.05,0.1,0.3$. The truncated horizon is $T=12$ and the time step is $\Delta t=0.1$. The learning rates are given by $\alpha_\psi(i)=\alpha_\xi(i)=\frac{0.001}{i^{0.61}}$. We plot in Figure \ref{fig:neural} the learnt value function, the mean and variance functions of the learnt policy by using Algorithm \ref{Alg:tracking} together with the neural networks after $N=100$ iterations of policy update. 
 \begin{figure}[h]
\centering
\centering
  \subfigure[]{
        \includegraphics[width=5cm]{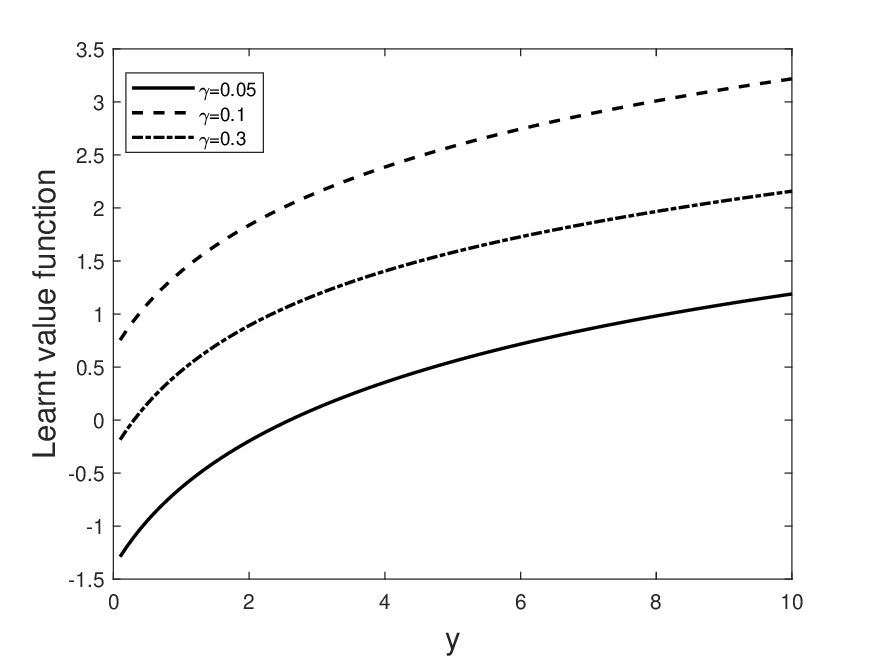}
    }\hspace{-5mm}
  \subfigure[]{
        \includegraphics[width=5cm]{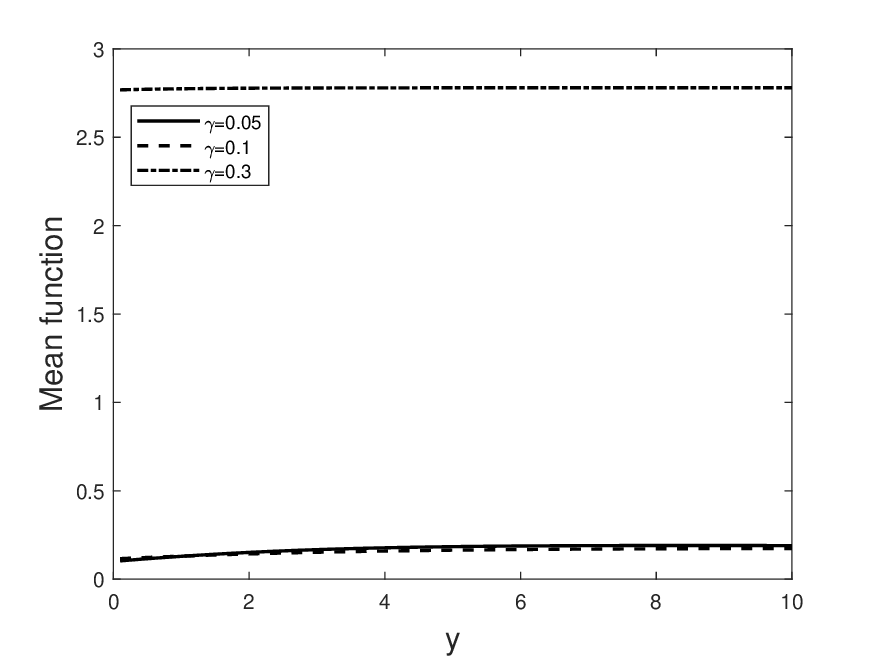}
    }\hspace{-5mm}
\subfigure[]{
        \includegraphics[width=5cm]{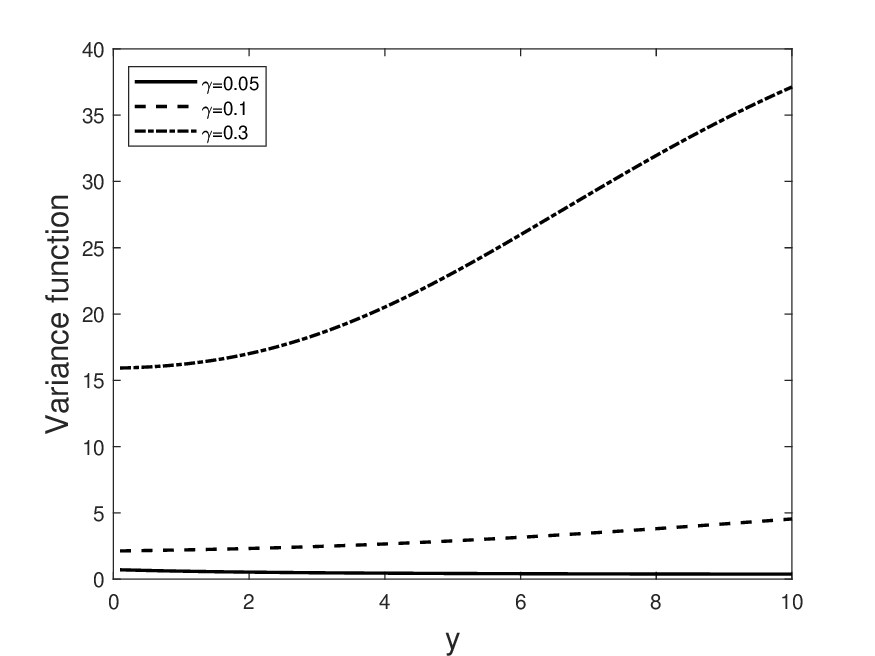}
    }
\caption{The learnt value function (left panel), mean function (middle panel) and variance function (right panel) of the learnt policy by implementing Algorithm \ref{Alg:tracking} together with the neural networks for $\gamma\neq \rho/d$.}\label{fig:neural}
\end{figure}

 We next also compare the performance of
the proposed q-learning algorithm with the classical stochastic control approach using maximum likelihood estimation (MLE) method for model parameters. In particular, we implement two methods based on real market data. We choose the S\&P 500 index as the benchmark process and select the Amazon.com Inc (AMZN) as the risky asset. We set the discount factor $\rho=0.1$, time step $\Delta t=1$ (Day) and initial wealth $\mathrm{v}=3281.1$ (10 higher then initial value of the benchmark process).  The daily returns from January 1, 2000 through July 30, 2020 are used as the training set while daily returns from July 31, 2020 through July 30, 2024 are used as the test set.

For the MLE method, we first use the training set to obtain the estimated values of the parameters as $\hat{\mu}=0.0012$, $\hat{\sigma}=0.0324$ and $\hat{\sigma}_Z=0.0126$. Suppose that the two Brownian motions $W^{\kappa}$ and $W^{\eta}$ are independent, i.e., $\kappa=1$. We then compute the optimal portfolio strategy given by \eqref{eq:theta} on the test set. On the other hand, for the reinforcement learning approach, we choose the temperature parameter $\gamma=\rho/d=0.1$, and train $N=500$ times with learning rates given by
\begin{align*}
 \alpha_{\psi}(i)=\left(\frac{0.005}{i^{0.71}},\frac{0.005}{i^{0.71}}\right), \quad \alpha_{\xi}=\frac{0.005}{i^{0.71}},\quad 1\leq i\leq 500  
\end{align*}
on the training set. We then implement the q-learning  algorithm on the test set with the learned value function and q-function. We plot the total wealth  (including capital injection) process of the agent and the cumulative capital injection process in Figure \ref{fig:injection}. It shows that the wealth process using the q-learning algorithm outperforms the one using the MLE method. The total injected capital required by the agent using the reinforcement learning method is $2243.46$, which is $5.86\%$ lower than the total  capital injection of $2383.21$ using the MLE method. These plots can show the effectiveness and robustness of the q-learning algorithm for some real market data.

\begin{figure}[h]
\centering
  \subfigure[]{
        \includegraphics[width=7.2cm]{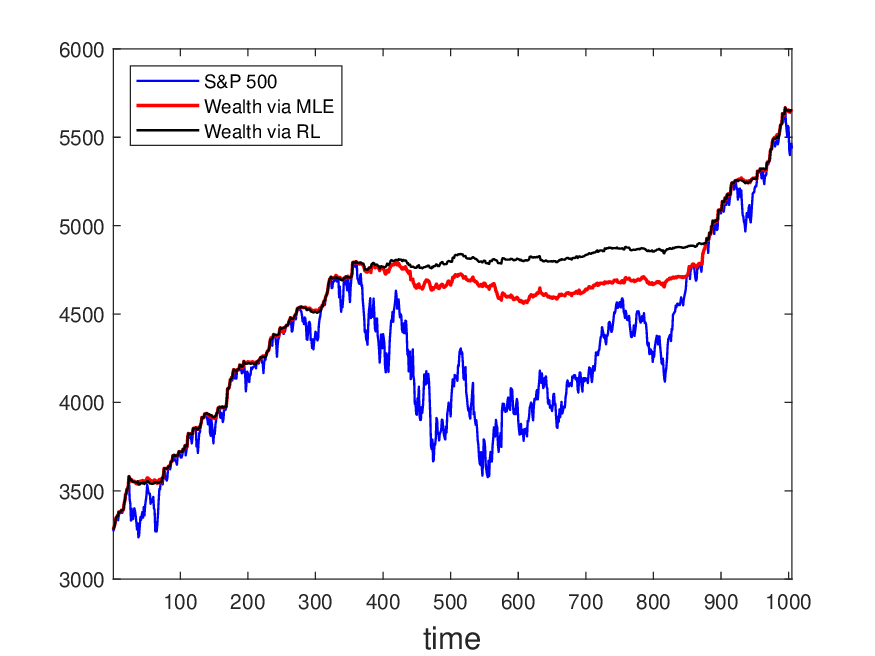}
    }\hspace{-8mm}
  \subfigure[]{
        \includegraphics[width=7.2cm]{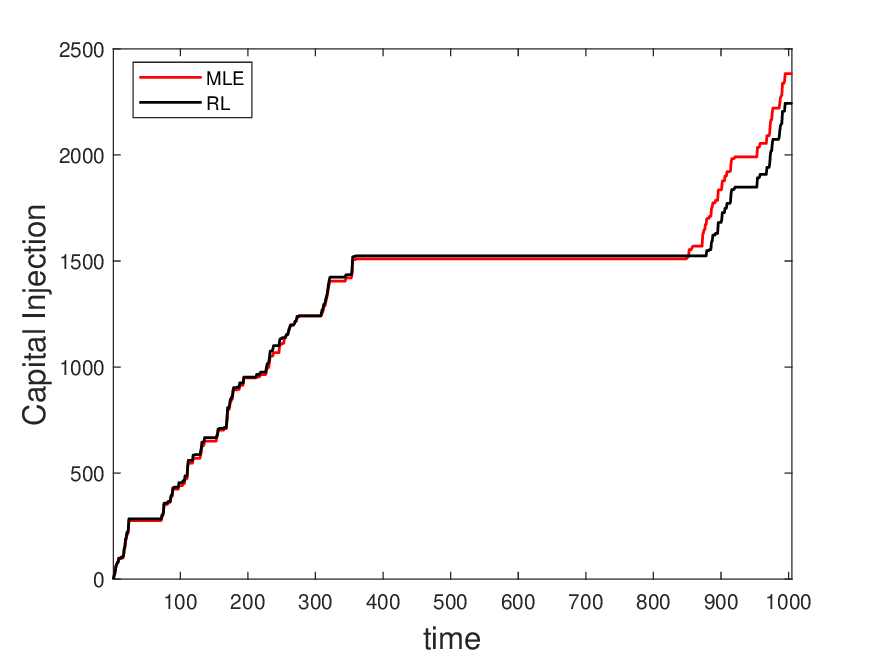}
    }
\caption{ 
 (a): The price process of the S\&P500, the total wealth (including capital injection) process under the MLE method and under the RL approach from  July 31, 2020 through July 30, 2024; (b): The cumulative capital injection process under the MLE method and under the RL approach from July 31, 2020 through July 30, 2024.}\label{fig:injection}
\end{figure}

\section{Proofs}\label{appendix:proof}

This section collects all proofs of the main results presented in previous sections.

\begin{proof}[Proof of Theorem \ref{thm:PIT}]
 We first prove item (i). For $T>0$ and $y\in\R_+$, applying It{\^o}'s formula to the process $e^{-\rho t}J(\tilde{Y}_t^{{\bm{\pi}}^{\prime}};\bm{\pi})$ from $0$ to $T$, we can obtain
\begin{align}\label{eq:Ito-J}
&e^{-\rho T}J(\tilde{Y}_T^{{\bm{\pi}}^{\prime}};\bm{\pi})-J(y;\bm{\pi})- \int_0^T e^{-\rho t} dL_t^{\bm{\pi}^{\prime}}-\gamma \int_0^T\int_{\cal A} e^{-\rho t}\ln\left(\bm{\pi}^{\prime}(\tilde{\theta};\tilde{Y}_t^{{\bm{\pi}}^{\prime}})\right)\bm{\pi}^{\prime}(\tilde{\theta};\tilde{Y}_t^{{\bm{\pi}}^{\prime}})d\tilde{\theta} dt\nonumber\\
&=\int_0^T e^{-\rho t}\int_{\cal A}\left(\left[H\left(\tilde{\theta},J_y(\tilde{Y}_t^{{\bm{\pi}}^{\prime}};\bm{\pi}), J_{yy}(\tilde{Y}_t^{{\bm{\pi}}^{\prime}};\bm{\pi})\right)-\rho J(\tilde{Y}_t^{{\bm{\pi}}^{\prime}};\bm{\pi})\right]+\frac{\sigma_Z^2(\tilde{Y}_t^{{\bm{\pi}}^{\prime}}+1)^2}{2}J_{yy}(\tilde{Y}_t^{{\bm{\pi}}^{\prime}};\bm{\pi})\right.\nonumber\\
&\quad\left.-\gamma \ln \bm{\pi}^{\prime}(\tilde{\theta};\tilde{Y}_t^{{\bm{\pi}}^{\prime}})\right) \bm{\pi}^{\prime}(\tilde{\theta};\tilde{Y}_t^{{\bm{\pi}}^{\prime}})d\tilde{\theta} dt+\int_0^T e^{-\rho t}\left(J_y(\tilde{Y}_t^{{\bm{\pi}}^{\prime}};\bm{\pi})-1\right) dL_t^{\bm{\pi}^{\prime}}\nonumber\\
&\quad+\int_0^Te^{-\rho t}\tilde{\sigma}(\tilde{Y}_t^{{\bm{\pi}}^{\prime}},\bm{\pi}^{\prime}(\cdot \mid \tilde{Y}_t^{{\bm{\pi}}^{\prime}})) J_y(\tilde{Y}_t^{{\bm{\pi}}^{\prime}};\bm{\pi})d B_t .
\end{align}
As $J(\cdot;\bm{\pi})\in C^2(\R_+)$ satisfies the PDE \eqref{RL}, we can see that
\begin{align}\label{eq:Neu-J}
\int_0^T e^{-\rho t}\left(J_y(\tilde{Y}_t^{{\bm{\bm{\pi}}}^{\prime}};\bm{\bm{\pi}})-1\right) dL_t^{\bm{\bm{\pi}}^{\prime}}=\int_0^T e^{-\rho t}\left(J_y(0;\bm{\bm{\pi}})-1\right) dL_t^{\bm{\bm{\pi}}^{\prime}}=0.
\end{align}
Moreover, it follows from Lemma 2 in \cite{JZ22c}  that, $\forall y\in\R_+$,
{\small\begin{align}\label{eq:HJB-J}
&\int_{\cal A}\left(\left[H\left(\tilde{\theta},J_y(y;\bm{\pi}), J_{yy}(y;\bm{\pi})\right)-\rho J(y;\bm{\pi})\right]+\frac{\sigma_Z^2(y+1)^2}{2}J_{yy}(y;\bm{\pi})-\gamma \ln \bm{\pi}^{\prime}(\tilde{\theta};y)\right) \bm{\pi}^{\prime}(\tilde{\theta};y)d\tilde{\theta} \\
&\geq \int_{\cal A}\left(\left[H\left(\tilde{\theta},J_y(y;\bm{\pi}), J_{yy}(y;\bm{\pi})\right)-\rho J(y;\bm{\pi})\right]+\frac{\sigma_Z^2(y+1)^2}{2}J_{yy}(y;\bm{\pi})-\gamma \ln \bm{\pi}(\tilde{\theta};y)\right) \bm{\pi}(\tilde{\theta};y)d\tilde{\theta}=0.\nonumber
\end{align}}
From \eqref{eq:Ito-J}, \eqref{eq:Neu-J} and \eqref{eq:HJB-J}, we deduce that
\begin{align}\label{eq:pi-J}
&e^{-\rho T}J(\tilde{Y}_T^{{\bm{\pi}}^{\prime}};\bm{\pi})-J(y;\bm{\pi})- \int_0^T e^{-\rho t} dL_t^{\bm{\pi}^{\prime}}-\gamma \int_0^T\int_{\cal A} e^{-\rho t}\ln\left(\bm{\pi}^{\prime}(\tilde{\theta};\tilde{Y}_t^{{\bm{\pi}}^{\prime}})\right)\bm{\pi}^{\prime}(\tilde{\theta};\tilde{Y}_t^{{\bm{\pi}}^{\prime}})d\tilde{\theta} dt\nonumber\\
&\qquad\geq \int_0^T e^{-\rho t}\tilde{\sigma}(\tilde{Y}_t^{{\bm{\pi}}^{\prime}},\bm{\pi}^{\prime}(\cdot \mid \tilde{Y}_t^{{\bm{\pi}}^{\prime}})) J_y(\tilde{Y}_t^{{\bm{\pi}}^{\prime}};\bm{\pi})d B_t .
\end{align}
Taking the expectation on both side of \eqref{eq:pi-J} and letting $T\to \infty$, by \eqref{eq:condition-pi} and the monotone convergence theorem (MCT), we conclude that for any $y\in\R_+$,
\begin{align*}
J(y;\bm{\pi})&\leq -\lim_{T\to \infty}  \Ex^{\mathbb{Q}^W}\left[\int_0^{T} e^{-\rho t} dL_t^{\bm{\pi}^{\prime}}\right]-\gamma\lim_{T\to \infty}  \Ex^{\mathbb{Q}^W}\left[ \int_0^{T}\int_{\cal A} e^{-\rho t}\ln\left(\bm{\pi}^{\prime}(\tilde{\theta};\tilde{Y}_t^{{\bm{\pi}}^{\prime}})\right)\bm{\pi}^{\prime}(\tilde{\theta};\tilde{Y}_t^{{\bm{\pi}}^{\prime}})d\tilde{\theta} dt\right]\nonumber\\
&\qquad+\lim_{T\to \infty}\Ex^{\mathbb{Q}^W}\left[e^{-\rho T}J(\tilde{Y}_T^{{\bm{\pi}}^{\prime}};\bm{\pi})\right] \nonumber\\
&=\Ex^{\mathbb{Q}^W}\left[- \int_0^{\infty} e^{-\rho t} dL_t^{\bm{\pi}^{\prime}}-\gamma \int_0^{\infty}\int_{\cal A} e^{-\rho t}\ln\left(\bm{\pi}^{\prime}(\tilde{\theta};\tilde{Y}_t^{{\bm{\pi}}^{\prime}})\right)\bm{\pi}^{\prime}(\tilde{\theta};\tilde{Y}_t^{{\bm{\pi}}^{\prime}})d\tilde{\theta} dt\right]\nonumber\\
&=J(y;\bm{\pi}^{\prime}).
\end{align*}
Next we deal with item (ii). Denote by $J^*(y)=J(y;\bm{\pi}^*)$. Thus, $J^*(y)$ satisfies Eq. \eqref{RL}, and Lemma 2 in \cite{JZ22c} shows that
\begin{align*}
&\int_{\cal A}\left(\left[H\left(\tilde{\theta},y,J_y^*(y), J_{yy}^*(y)\right)-\rho J^*(y)\right]+\frac{\sigma_Z^2(y+1)^2}{2}J_{yy}^*(y)-\gamma \ln \bm{\pi}^*(\tilde{\theta};y)\right) \bm{\pi}^{*}(\tilde{\theta};y)d\tilde{\theta} \nonumber\\
&= \sup_{\bm{\pi}\in{\mathcal{ P}}({\mathcal {A}})} \int_{\cal A}\left(\left[H\left(\tilde{\theta},y,J_y^*(y), J_{yy}^*(y)\right)-\rho J^*(y)\right]+\frac{\sigma_Z^2(y+1)^2}{2}J^*_{yy}(y)-\gamma \ln \bm{\pi}(\tilde{\theta};y)\right) \bm{\pi}(\tilde{\theta};y)d\tilde{\theta}=0.
\end{align*}
Fix $(T,y)\in(0,\infty)\times\R_+$ and $\bm{\pi}\in\Pi$. By applying It\^{o}'s formula to $e^{-\rho T}J^*(\tilde{Y}_{T}^{\bm{\pi}^*})$, we arrive at
\begin{align}\label{eq:itoveri}
e^{-\rho T}J^*(\tilde{Y}_{T}^{\bm{\pi}})&=J^*(y)+\int_0^{T} e^{-\rho s} J^*_y(\tilde{Y}_{s}^{\bm{\pi}})dL_s^Y+\int_0^{T} e^{-\rho s}({\cal L}^{\bm{\pi}(\cdot|\tilde{Y}_s^{\bm{\pi}})} J^*-\rho J^*)(\tilde{Y}_{s}^{\bm{\pi}})ds\nonumber\\
&\quad+\int_0^{T}e^{-\rho s} J^*_y(\tilde{Y}_{s}^{\bm{\pi}})\tilde{\sigma}(\tilde{Y}_s^{\bm{\pi}},\bm{\pi}(\cdot|\tilde{Y}_s^{\bm{\pi}})) dB_s,
\end{align}
where, for any  $\bm{\pi}\in\Pi$, the operator ${\cal L}^{\bm{\pi}}$ acting on $C^2(\R_+)$ is defined by
\begin{align*}
{\cal L}^{\bm{\pi}}g(y)&:=\tilde{b}(y,\bm{\pi})g'(y)+\frac{1}{2}\tilde{\sigma}(y,\bm{\pi})g''(y),\quad \forall y\in\R_+.
\end{align*}
Taking the expectation on both sides of \eqref{eq:itoveri}, we deduce from the Neumann boundary condition $J^*_y(0)=1$ that
\begin{align}\label{eq:value-ineq}
\Ex^{\Qx^W}\left[- \int_0^{T} e^{-\rho s}dL_s^Y \right]&=J^*(y)-\Ex^{\Qx^W}\left[e^{-\rho T}v(\tilde{Y}_{T}^{\bm{\pi}})\right]+\Ex^{\Qx^W}\left[\int_0^{T} e^{-\rho s}({\cal L}^{\bm{\pi}(\cdot|\tilde{Y}_s^{\bm{\pi}})} v-\rho v)(\tilde{Y}_{s})ds\right]\nonumber\\
&\leq J^*(y)-\Ex^{\Qx^W}\left[e^{-\rho T}J^*\left(\tilde{Y}_{T}^{\bm{\pi}}\right)\right].
\end{align}
Here, the last inequality in \eqref{eq:value-ineq} holds true due to $({\cal L}^{\bm{\pi}} v-\rho v)(y)\leq 0$ for all $y\in\R_+$ and $\bm{\pi}\in{\cal P}({\cal A})$. Toward this end, letting $T\to\infty$ in \eqref{eq:value-ineq}, we obtain from  \eqref{eq:condition-optimal-pi} and MCT that, for all $\bm{\pi}\in \Pi$,
\begin{align*}
\Ex^{\Qx^W}\left[- \int_0^{\infty} e^{-\rho s}dL_s^Y \right]\leq v(y),\quad \forall y\in\R_+,
\end{align*}
where the equality holds when $\bm{\pi}=\bm{\pi}^*$. This implies that $J^*(y)=v(y)$ is the optimal value function and $\bm{\pi}^*$ is the optimal policy.

Next, we move to item (iii). To find the fix point, we use the iteration method and start with a special Gaussian policy: $\bm{\pi}_0(\cdot|y)\sim \mathcal{N}\left(\cdot~|c_1(1+y),c_2(1+y)^2\right)$, with vector $c_1\in\R^d$ and  positive matrix $c_2\in \R^{d\times d}$.  Then it follows that the resulting value function $J^{\bm{\pi}_0}(y)$ satisfies the equation:
\begin{align}\label{eq:PDE-pi0}
\begin{cases}
\displaystyle \int_{\mathcal {A}}\left[H(\theta,y,J^{\bm{\pi}_0}_y(y),J^{\bm{\pi}_0}_{yy}(y))
-\gamma \log\bm{\pi}_0(\theta|y)\right]\bm{\pi}_0(\theta|y)d\theta +\frac{\sigma_Z^2}{2}(y+1)^2J^{\bm{\pi}_0}_{yy}(y)=\rho J^{\bm{\pi}_0}(y),\\[1.2em]
\displaystyle J^{\bm{\pi}_0}_y(y)=1.
\end{cases}
\end{align}
By standard arguments, we can show that the classical solution of \eqref{eq:PDE-pi0} is given by
\begin{align*}
J^{\bm{\pi}_0}(y)=\ln(1+y)+C,\quad \forall y\in\R_+,
\end{align*}
where $C$ is a constant defined by
\begin{align*}
C:=\frac{1}{\rho}\left(\mu^{\top}c_1+\sqrt{1-\kappa^2}\sigma_Zc_1^{\top}\sigma\eta-\frac{1}{2}c_1^{\top}\sigma\sigma^{\top}c_1-\frac{1}{2}\text{tr}(\sigma\sigma^{\top}c_2)-\frac{1}{2}\sigma_Z^2+\frac{1}{2}\gamma\ln\left((2\pi)^d|c_2|\right)+\frac{d}{2}\gamma\right),
\end{align*}
and for a matrix $A\in\R^{d\times d}$, we denote $|A|:=\text{det}(A)$. Then, using once iteration, we get that
\begin{align}\label{eq:fix-point}
\bm{\pi}_1(\cdot|y)={\cal I}(\bm{\pi}_0) \sim  \mathcal{N}\left(\cdot~\Big|(\sigma\sigma^{\top})^{-1}(\mu+\sigma_z\sqrt{1-\kappa^2}\sigma \eta)(1+y),(\sigma\sigma^{\top})^{-1}\gamma(1+y)^2\right).
\end{align}
Again, we can calculate the corresponding reward function as 
\begin{align}\label{eq:fix-point-function}
J^{\pi_1}(y)=\ln(1+y)+\frac{1}{\rho}\left(-\frac{1}{2}\sigma_Z^2\kappa^2+\sqrt{1-\kappa^2}\zeta+\frac{\gamma}{2}\ln\left(\frac{(2\pi \gamma)^d}{|\sigma\sigma^{\top}|}\right)+\alpha\right)
\end{align}
Then the  iteration is applicable again, which yields the improved policy $\bm{\pi}_2$ as exactly the Gaussian policy $\bm{\pi}_1$ given in \eqref{eq:fix-point}, together with the reward function in \eqref{eq:fix-point-function}. Then we find that ${\bm{\pi}}^*$ give in \eqref{eq:fixpoint-pi} is a fixed point of ${\cal I}$. Thus, we complete the proof.
\end{proof}

\begin{proof}[Proof of Proposition \ref{prop:martingale-condition}]
For $T>0$ and $y\in\R_+$, applying It{\^o}'s formula to the process $e^{-\rho s}J(Y_s^{\bm\pi};{\bm \pi})$ from $0$ to $t$, we can obtain that
\begin{align}\label{eq:Ito-JJ}
&e^{-\rho s}J(Y_s^{{\bm \pi}};{\bm \pi})-J(y;{\bm \pi})-\int_0^t e^{-\rho s}\hat{q}\left( Y_{s}^{{\bm \pi}},\tilde{\theta}_{s}^{{\bm \pi}}\right)d s-\int_0^t e^{-\rho s}dL_{s}^{{\bm \pi}}\nonumber\\
&=\int_0^t e^{-\rho s}\left(H\left(\tilde{\theta}^{{\bm \pi}},Y_s^{{\bm \pi}},J_y(Y_s^{{\bm \pi}};{\bm \pi}), J_{yy}(Y_s^{{\bm \pi}};{\bm \pi})\right)-\rho J(Y_s^{{\bm \pi}};{\bm \pi})+\frac{\sigma_Z^2(Y_s^{{\bm \pi}}+1)^2}{2}J_{yy}(Y_s^{{\bm \pi}};{\bm \pi})\right)ds\nonumber\\
&\quad+\int_0^t e^{-\rho t}\left(J_y(Y_s^{{\bm \pi}};{\bm \pi})-1\right) dL_s^{{\bm \pi}}+\int_0^t e^{-\rho s}(\tilde{\theta}^{{\bm \pi}})^{\top}\sigma J_y(Y_s^{{\bm \pi}};{\bm \pi})d W_s\nonumber\\
&\quad -\int_0^t  e^{-\rho s}\sigma_Z(Y_s^{{\bm \pi}}+1)J_y(Y_s^{{\bm \pi}};{\bm \pi})d\tilde{W}^{\kappa}_s\nonumber\\
&=\int_0^t e^{-\rho s}\left(q(Y_s^{{\bm \pi}},\tilde{\theta}^{{\bm \pi}};{\bm \pi})-\hat{q}\left( Y_{s}^{{\bm \pi}},\tilde{\theta}_{s}^{{\bm \pi}}\right)\right) ds+\int_0^t e^{-\rho s}(\tilde{\theta}^{{\bm \pi}})^{\top}\sigma J_y(Y_s^{{\bm \pi}};{\bm \pi})d W_s \nonumber\\
&\quad-\int_0^t  e^{-\rho s}\sigma_Z(Y_s^{{\bm \pi}}+1)J_y(Y_s^{{\bm \pi}};{\bm \pi})d\tilde{W}^{\kappa}_s.
\end{align}
From Eq.~\eqref{eq:Ito-JJ}, we can see that, if $\hat{q}(y,\tilde{\theta})=q(y,\tilde{\theta}; {\bm \pi})$  for all $(y,\tilde{\theta}) \in\R_+\times {\cal A}$, then
the above process, and hence the process defined by \eqref{eq:martinglae-J} is an  $((\mathcal{F}_t)_{t \geq 0},\mathbb{Q})$-martingale.

On the other hand,  if the process \eqref{eq:martinglae-J} is an  $((\mathcal{F}_t)_{t \geq 0},\mathbb{Q})$-martingale, we next show that $\hat{q}(y,\tilde{\theta})=q(y,\tilde{\theta}; {\bm \pi})$ for $(y,\tilde{\theta}) \in\R_+\times {\cal A}$. By \eqref{eq:Ito-JJ}, $\int_0^t e^{-\rho s}(q(Y_s^{{\bm \pi}},\tilde{\theta}^{{\bm \pi}};{\bm \pi})-\hat{q}( Y_{s}^{{\bm \pi}},\tilde{\theta}_{s}^{{\bm \pi}})) ds$ for $t\geq0$ is a continuous local martingale with finite variation and hence zero quadratic variation. Hence, it follows that $\int_0^t e^{-\tilde{\rho}s}(q(Y_s^{{\bm \pi}},\tilde{\theta}^{{\bm \pi}};{\bm \pi})-\hat{q}( Y_{s}^{{\bm \pi}},\tilde{\theta}_{s}^{{\bm \pi}})) ds=0$ for all $t \in[0, \infty)$, $\mathbb{Q}$-a.s. (see Chapter 1, Exercise 5.21 in \cite{Karatzas14}). For $y\in\R_+, a\in{\cal A}$, denote by $h(y,\tilde{\theta})=q(y,\tilde{\theta};{\bm \pi})-\hat{q}(y,\tilde{\theta})$, which is a continuous function that maps $\R_+ \times \mathcal{A}$ to $\R_+$. Next,  we argue by contradiction. Suppose there exists a pair $(y_0,\tilde{\theta}_0)\in\R_+\times{\cal A}$ and $\epsilon>0$ such that $h(y_0,\tilde{\theta}_0)>\epsilon$. By the continuity of  $h$, there exists $\delta>0$ such that $h(y,\tilde{\theta})>\epsilon/2$ for all $(y,\tilde{\theta})$ with $\max\{|y-y_0| ,|\tilde{\theta}-\tilde{\theta}_0|\}<\delta$. 

Now, consider the state process, still denoted by $Y^{\bm \pi}$, starting from $(y_0,\tilde{\theta}_0)$, namely, $(Y_t^{\bm \pi})_{t\geq0}$ follows (6) with $Y^{\bm \pi}_0=y_0$ and $\tilde{\theta}_0^{\bm \pi}=\tilde{\theta}_0$. Define the stopping time $\tau$ by $\tau:=\inf \left\{t\geq 0:\left|Y_t^{\bm \pi}-y_0\right|>\delta\right\}$. We have already shown that there exists $\Omega_0 \in \mathcal{F}$ with $\mathbb{Q}\left(\Omega_0\right)=0$ such that for all $\omega \in \Omega \backslash \Omega_0$,  $\int_0^t e^{-\rho s}(q(Y_s^{{\bm \pi}},\tilde{\theta}^{{\bm \pi}};{\bm \pi})-\hat{q}( Y_{s}^{{\bm \pi}},\tilde{\theta}_{s}^{{\bm \pi}})) ds=0$ for all $t \in[0, \infty)$. It follows from Lebesgue's differentiation theorem that, for any $\omega \in \Omega \backslash \Omega_0$, $h(Y_t^{\bm \pi}(\omega),\tilde{\theta}_t^{\bm \pi}(\omega))=0,~ \text {a.e., for }~ t \in[0, \tau(\omega)]$. Consider the set ${\cal O}(\omega)=\{t \in[0, \tau(\omega)]: |\tilde{\theta}_t^{\bm \pi}(\omega)-\tilde{\theta}_0|<\delta\} \subset [0, \tau(\omega)]$. Because $h( Y_t^{\bm \pi}(\omega), \tilde{\theta}_t^{\bm \pi}(\omega))>\frac{\epsilon}{2}$ when $t \in {\cal O}(\omega)$, we conclude that ${\cal O}(\omega)$ has Lebesgue measure zero for any $\omega \in \Omega \backslash \Omega_0$. That is
\begin{align*}
\int_0^{\infty} {\bf 1}_{\{t \leq \tau(\omega)\}} {\bf 1}_{\left\{|\tilde{\theta}_t^{\bm \pi}(\omega)-\tilde{\theta}_0|<\delta\right\}} dt=0 .
\end{align*}
Integrating $\omega$ with respect to $\mathbb{Q}$ and applying Fubini's theorem, we obtain that
\begin{align*}
0 & =\int_{\Omega} \int_0^{\infty}  {\bf 1}_{\{t \leq \tau(\omega)\}} {\bf 1}_{\left\{|\tilde{\theta}_t^{\bm \pi}(\omega)-\tilde{\theta}_0|<\delta\right\}} dt \mathbb{Q}(d\omega)\int_0^{\infty}  \Ex^{\Qx}\left[{\bf 1}_{\{t \leq \tau\}} {\bf 1}_{\left\{|\tilde{\theta}_t^{\bm \pi}-\tilde{\theta}_0|<\delta\right\}} \right] dt\\
& =\int_{0}^\infty \mathbb{E}\left[{\bf 1}_{\{t \leq \tau\}} \mathbb{Q}^{\Qx}\left(|\tilde{\theta}_t^{\bm \pi}-\tilde{\theta}_0|<\delta|\mathcal{F}_t\right)\right] dt=\int_{0}^\infty \mathbb{E}\left[{\bf 1}_{\{t \leq \tau\}}\int_{|\tilde{\theta}-\tilde{\theta}_0|<\delta} {\bm \pi}\left(\tilde{\theta}| Y_t^{{\bm \pi}}\right)d\tilde{\theta}\right]dt \\
& \geq \min _{\left|y-y_0\right|<\delta}\left\{\int_{|\tilde{\theta}-\tilde{\theta}_0|<\delta} {\bm \pi}\left(\tilde{\theta}| y\right)d\tilde{\theta}\right\} \int_{0}^\infty \mathbb{E}^{\Qx}\left[{\bf 1}_{\{t \leq \tau\}}\right] dt \\
& =\min _{\left|y-y_0\right|<\delta}\left\{\int_{|\tilde{\theta}-\tilde{\theta}_0|<\delta} {\bm \pi}\left(\tilde{\theta}| y\right)d\tilde{\theta}\right\} \mathbb{E}^{\Qx}[\tau]\geq 0 .
\end{align*}
The above implies that $\min _{|y-y_0|<\delta}\{\int_{|\tilde{\theta}-\tilde{\theta}_0|<\delta} {\bm \pi}(\tilde{\theta}| y)d\tilde{\theta}\}=0$. However, this contradicts Definition \ref{def:admissible-pi} about an admissible policy. Indeed, Definition \ref{def:admissible-pi}-(i) stipulates $\operatorname{supp} {\bm \pi}(\cdot| y)=\mathcal{A}$ for any $y\in\R_+$; hence $\int_{|\tilde{\theta}-\tilde{\theta}_0|<\delta} {\bm \pi}(\tilde{\theta}| y)d\tilde{\theta}>0$. Then, the continuity in Definition \ref{def:admissible-pi}-(iii) yields $\min _{|y-y_0|<\delta}\{\int_{|\tilde{\theta}-\tilde{\theta}_0|<\delta} {\bm \pi}(\tilde{\theta}| y)d\tilde{\theta}\}>0$, which is a contradiction. Hence, we conclude that $q(y,\tilde{\theta};{\bm \pi})=\hat{q}(y,\tilde{\theta})$ for every $(y,\tilde{\theta})\in\R_+\times {\cal A}$.
\end{proof}

\begin{proof}[Proof of Theorem \ref{thm:martingale-condition}]
If $\hat{J}$ and $\hat{q}$ are respectively the value function and the $q$-function associated with the policy ${\bm \pi}$, it follows from the same argument as in the proof of Proposition \ref{prop:martingale-condition} that
\begin{align*}
e^{-\rho t} \hat{J}( Y_t^{\bm\pi})-\int_0^t e^{-\rho s}\hat{q}\left( Y_{s}^{\bm\pi},\tilde{\theta}_{s}^{\bm \pi}\right)d s-\int_0^t e^{-\rho s}dL_{s}^{\bm \pi},\quad t\geq0
\end{align*}
is an $((\mathcal{F}_t)_{t \geq 0}, \mathbb{Q})$-martingale.
On the other hand,  assume that $e^{-\rho t} \hat{J}( Y_t^{{\bm \pi}})-\int_0^t e^{-\rho s}\hat{q}( Y_{s}^{{\bm \pi}},\tilde{\theta}_{s}^{{\bm \pi}})d s-\int_0^t e^{-\rho s}dL_{s}^{{\bm \pi}}$ is an  $((\mathcal{F}_t)_{t \geq 0}, \mathbb{Q})$ martingale. It then holds that, for any $T>0$,
\begin{align}\label{eq:J-Q}
\hat{J}(y)=\Ex^{\Qx}\left[e^{-\rho T} \hat{J}( Y_T^{\bm\pi})-\int_0^T e^{-\rho s}\hat{q}\left( Y_{s}^{\bm\pi},\tilde{\theta}_{s}^{\bm \pi}\right)d s-\int_0^T e^{-\rho s}dL_{s}^{\bm \pi}\right].
\end{align}
Integrating over the action randomization with respect to the policy ${\bm \pi}$ on the side of Eq. \eqref{eq:J-Q}, we get that
\begin{align}\label{eq:J-QW}
\hat{J}(y,{\bm \pi})=\Ex^{\Qx^W}\left[e^{-\rho T} \hat{J}( \tilde{Y}_T^{\bm\pi};{\bm \pi})-\int_0^T e^{-\rho s}\int_{\cal A}\hat{q}\left( \tilde{Y}_{s}^{\bm\pi},\tilde{\theta}\right){\bm \pi}(\tilde{\theta}|\tilde{Y}_s^{\bm \pi})d\tilde{\theta}d s-\int_0^T e^{-\rho s}dL_{s}^{\bm \pi}\right].
\end{align}
In view that $\int_{\mathcal{A}}[\hat{q}(y,\tilde{\theta})-\gamma \ln{\bm \pi}(\tilde{\theta}|y)] {\bm \pi} (\tilde{\theta}| y)d\tilde{\theta}=0$,  we obtain 
\begin{align}\label{eq:J-QW-pi}
\hat{J}(y)=\Ex^{\Qx^W}\left[e^{-\rho T} \hat{J}( \tilde{Y}_T^{\bm\pi})-\gamma\int_0^T e^{-\rho s}\int_{\cal A}\ln{\bm \pi}(\tilde{\theta}|\tilde{Y}_s^{\bm \pi}){\bm \pi}(\tilde{\theta}|\tilde{Y}_s^{\bm \pi})d\tilde{\theta}d s-\int_0^T e^{-\rho s}dL_{s}^{\bm \pi}\right].
\end{align}
Letting $T$ go to infinity on both side of \eqref{eq:J-QW-pi}, we deduce from MCT and \eqref{eq:condition-J} that
\begin{align}\label{eq:J-QW-pi-2}
\hat{J}(y)&=\Ex^{\Qx^W}\left[-\gamma\int_0^{\infty} e^{-\rho s}\int_{\cal A}\ln{\bm \pi}(\tilde{\theta}|\tilde{Y}_s^{\bm \pi}){\bm \pi}(\tilde{\theta}|\tilde{Y}_s^{\bm \pi})d\tilde{\theta}d s-\int_0^{\infty} e^{-\rho s}dL_{s}^{\bm \pi}\right].
\end{align}
This yields that $\hat{J}(y)=J(y;{\bm \pi})$ for all $y\in\R_+$ by virtue of  \eqref{RL}. Furthermore, based on Proposition \ref{prop:martingale-condition}, the martingale condition implies that $\hat{q}(y,\tilde{\theta})=q(y,\tilde{\theta} ; {\bm \pi})$ for all $(y,\tilde{\theta})\in\R_+\times {\cal A}$.

Finally, if ${\bm \pi}(\tilde{\theta}| y)=\frac{\exp\{\frac{1}{\gamma} \hat{q}(y,\tilde{\theta})\}}{\int_{\mathcal{A}} \exp\{\frac{1}{\gamma} \hat{q}(y,\tilde{\theta})\} d\tilde{\theta}}$ satisfies the condition \eqref{eq:condition-pi-geq}, then ${\bm \pi}={\cal I}({\bm \pi})$ where $\mathcal{I}$ is the map defined in Theorem \ref{thm:PIT}. This in turn implies that ${\bm \pi}$ is the optimal policy and $\hat{J}$ is the optimal value function.
\end{proof}

\begin{proof}[Proof of Theorem \ref{thm:convergence-policy}]
We first state the following auxiliary lemma.
\begin{lemma}[Lemma 8 in \cite{JZ22b}]\label{lem:convergence}
Let $f_h(x)=f(x)+r_h(x)$, where $f$ is a continuous function and $r_h$ converges to 0 uniformly on any compact set as $h \rightarrow 0$. Assume that $f_h(x_h^*)=0$ and $\lim _{h \rightarrow 0} x_h^*=x^*$. Then $f(x^*)=0$. 
\end{lemma}
We first apply Lemma \ref{lem:convergence} to prove item (i). We can take $f(x)=\Ex^{\Qx}[\int_0^{\infty}\varsigma_tdM_t^{\xi,\psi}]$ and  $f_h(x)=\Ex^{\Qx}[\int_0^{T}\varsigma_tdM_t^{\xi,\psi}]$ with $h=1/T$, $x=(\xi,\psi)$. Thus, we need to prove that
\begin{align}\label{eq:convergenct-T-1}
\Ex^{\Qx}\left[\int_0^{\infty}\varsigma_tdM_t^{\xi,\psi}\right]-\Ex^{\Qx}\left[\int_0^{T}\varsigma_tdM_t^{\xi,\psi}\right]=\Ex^{\Qx}\left[\int_T^{\infty}\varsigma_tdM_t^{\xi,\psi}\right]\to 0, 
\end{align}
as $T\to \infty$ uniformly on any compact subset of $\Theta\times \Psi$.
Let ${\cal C}\subset\Theta\times \Psi$ be a compact set, then ${\cal C}$ is a bounded closed set.  By Proposition \ref{prop:orthogonality} and  assumptions {\bf(A$_{\xi}$)} and {\bf(A$_{\psi}$)}, we have
\begin{align}\label{eq:convergenct-T}
&\Ex^{\Qx}\left[\int_T^{\infty}\varsigma_tdM_t^{\xi,\psi}\right]\nonumber\\
&= \Ex^{\Qx}\left[\int_T^{\infty}e^{-\rho t}\varsigma_t\left(  {\cal L}^{\tilde{\theta}^{\bm \pi}_t}J^{\xi}(Y_t^{{\bm \pi}})-\rho  J^{\xi}(Y_t^{\bm \pi})-q^{\psi}( Y_{t}^{\bm \pi},\tilde{\theta}_{t}^{\bm \pi})\right)dt\right]+\Ex^{\Qx}\left[\int_T^{\infty}e^{-\rho t}\varsigma_tdL_t^{\bm \pi}\right]\nonumber\\
&=\Ex^{\Qx}\left[\int_T^{\infty}e^{-\rho t}\varsigma_t\left((\tilde{\theta}^{\bm \pi}_t)^{\top}\mu J_y^{\xi}(Y_t^{\bm \pi})+\frac{1}{2}(\tilde{\theta}^{\bm \pi}_t)^{\top}\sigma\sigma^{\top}\tilde{\theta}^{\bm \pi}_tJ_{yy}^{\xi}(Y_t^{\bm \pi})-\sqrt{1-\kappa^2}\sigma_Z(Y_t^{\bm \pi}+1)(\tilde{\theta}^{\bm \pi}_t)^{\top}\sigma\eta J_{yy}^{\xi}(Y_t^{\bm \pi})\right.\right.\nonumber\\
&\qquad+\left.\left.\frac{1}{2}\sigma_Z^2 (Y_t^{\bm \pi}+1)^2 J_{yy}^{\xi}(Y_t^{\bm \pi})-\rho  J^{\xi}(Y_t^{\bm \pi})-q^{\psi}( Y_{t}^{\bm \pi},\tilde{\theta}_{t}^{\bm \pi})\right)dt\right]+\Ex^{\Qx}\left[\int_T^{\infty}e^{-\rho t}\varsigma_tdL_t^{\bm \pi}\right]\nonumber\\
&\leq (G_J(\xi)+G_q(\psi))\Ex^{\Qx}\left[\int_T^{\infty}e^{-\rho t}\varsigma_t\left(\left(\left|(\tilde{\theta}^{\bm \pi}_t)^{\top}\mu \right|+\frac{1}{2}(\tilde{\theta}^{\bm \pi}_t)^{\top}\sigma\sigma^{\top}\tilde{\theta}^{\bm \pi}_t+\left|\sqrt{1-\kappa^2}\sigma_Z(Y_t^{\bm \pi}+1)(\tilde{\theta}^{\bm \pi}_t)^{\top}\sigma\eta\right|\right.\right.\right.\nonumber\\
&\qquad+\left.\left.\left.\frac{1}{2}\sigma_Z^2 (Y_t^{\bm \pi}+1)^2+\rho \right) \tilde{J}(Y_t^{\bm \pi})+\tilde{q}( Y_{t}^{\bm \pi},\tilde{\theta}_{t}^{\bm \pi})\right)dt\right]+\Ex^{\Qx}\left[\int_T^{\infty}e^{-\rho t}\varsigma_tdL_t^{\bm \pi}\right].
\end{align}
This yields that as $T\to \infty$, $\Ex^{\Qx}[\int_T^{\infty}\varsigma_tdM_t^{\xi,\psi}]$ converges to 0 uniformly in $(\xi,\psi)\in {\cal C}$. By Lemma \ref{lem:convergence}, we get the desired result.

Next we deal with item (ii). By Lemma \ref{lem:convergence} again, it is sufficient to show that
\begin{align}\label{eq:convergenct-delta}
\Ex^{\Qx}\left[\int_0^{\infty}\varsigma_tdM_t^{\xi,\psi}\right]-\Ex^{\Qx}\left[\sum_{k=0}^{K-1}\varsigma_{t_k}\left(M_{t_{k+1}}^{\xi,\psi}-M_{t_k}^{\xi,\psi}\right)\right]\to 0, 
\end{align}
as $\Delta t\to \infty$ uniformly on any compact subset of $\Theta\times \Psi$. Let ${\cal C}\subset\Theta\times \Psi$ be a compact set, which is also a bounded closed set. It follows from that
\begin{align*}
&\Ex^{\Qx}\left[\int_0^{\infty}\varsigma_tdM_t^{\xi,\psi}\right]-\Ex^{\Qx}\left[\sum_{k=0}^{K-1}\varsigma_{t_k}\left(M_{t_{k+1}}^{\xi,\psi}-M_{t_k}^{\xi,\psi}\right)\right]\nonumber\\
& =\Ex^{\Qx}\left[\int_0^{\infty}\varsigma_tdM_t^{\xi,\psi}\right]-\Ex^{\Qx}\left[\sum_{k=0}^{K-1}\int_{t_k}^{t_{k+1}}\varsigma_{t_k}dM_t^{\xi,\psi}\right]\nonumber\\
&=\Ex^{\Qx}\left[\sum_{k=0}^{K-1}\int_{t_k}^{t_{k+1}}e^{-\rho t}(\varsigma_t-\varsigma_{t_k})\left(  {\cal L}^{\tilde{\theta}^{\bm \pi}_t}J^{\xi}(Y_t^{{\bm \pi}})-\rho  J^{\xi}(Y_t^{\bm \pi})-q^{\psi}( Y_{t}^{\bm \pi},\tilde{\theta}_{t}^{\bm \pi})\right)dt\right]\nonumber\\
&\quad+\Ex^{\Qx}\left[\int_0^{T}e^{-\rho t}\varsigma_tdL_t^{\bm \pi}-\sum_{k=0}^{K-1}\int_{t_k}^{t_{k+1}}e^{-\rho t}\varsigma_{t_k}dL_t^{\bm \pi}\right].\nonumber\\
\end{align*}
Note that
\begin{align*}
&\Ex^{\Qx}\left[\int_0^{T}e^{-\rho t}\varsigma_tdL_t^{\bm \pi}-\sum_{k=0}^{K-1}\int_{t_k}^{t_{k+1}}e^{-\rho t}\varsigma_{t_k}dL_t^{\bm \pi}\right]\\
&\quad=\Ex^{\Qx}\left[\int_0^{T}e^{-\rho t}\varsigma_tdL_t^{\bm \pi}-\sum_{k=0}^{K-1}e^{-\rho t_k}\varsigma_{t_k}\left(L_{t_{k+1}}^{\bm \pi}-L_{t_k}^{\bm \pi}\right)\right]+\Ex^{\Qx}\left[\sum_{k=0}^{K-1}\int_{t_k}^{t_{k+1}}(e^{-\rho t}-e^{-\rho t_k})\varsigma_{t_k}dL_t^{\bm \pi}\right]\\
&\quad\leq \Ex^{\Qx}\left[\int_0^{T}e^{-\rho t}\varsigma_tdL_t^{\bm \pi}-\sum_{k=0}^{K-1}e^{-\rho t_k}\varsigma_{t_k}\left(L_{t_{k+1}}^{\bm \pi}-L_{t_k}^{\bm \pi}\right)\right]\nonumber\\
&\qquad+\sup_{|t-s|\leq \Delta t,t,s\in[0,T]}(e^{-\rho t}-e^{-\rho s})\Ex^{\Qx}\left[\sum_{k=0}^{K-1}\int_{t_k}^{t_{k+1}}\varsigma_{t_k}dL_t^{\bm \pi}\right],
\end{align*}
where the 1st term in above inequality converges to 0 as $\Delta t\to 0$ by using the definition of Riemann-Stieltjes integral and the 2nd term converges to 0 as $\Delta t\to 0$ since $\sup_{|t-s|\leq \Delta t,t,s\in[0,T]}(e^{-\rho t}-e^{-\rho s})\to 0$. 
%&\leq \sum_{k=0}^{K-1}\int_{t_k}^{t_{k+1}}\Ex^{\Qx}\left[(\varsigma_t-\varsigma_{t_k})^2\right]^{\frac{1}{2}}\Ex^{\Qx}\left[\left(  {\cal L}^{\tilde{\theta}^{\bm \pi}_t}J^{\xi}(Y_t^{{\bm \pi}})-\rho  J^{\xi}(Y_t^{\bm \pi})-q^{\psi}( Y_{t}^{\bm \pi},\tilde{\theta}_{t}^{\bm \pi})\right)^2\right]^{\frac{1}{2}}dt\nonumber\\
On the other hand, we have
\begin{align*}
&\Ex^{\Qx}\left[\sum_{k=0}^{K-1}\int_{t_k}^{t_{k+1}}e^{-\rho t}(\varsigma_t-\varsigma_{t_k})\left(  {\cal L}^{\tilde{\theta}^{\bm \pi}_t}J^{\xi}(Y_t^{{\bm \pi}})-\rho  J^{\xi}(Y_t^{\bm \pi})-q^{\psi}( Y_{t}^{\bm \pi},\tilde{\theta}_{t}^{\bm \pi})\right)dt\right]\\
&\leq\sum_{k=0}^{K-1}\left(\int_{t_k}^{t_{k+1}}\Ex^{\Qx}\left[(\varsigma_t-\varsigma_{t_k})^2\right]dt\right)^{\frac{1}{2}}\left(\int_{t_k}^{t_{k+1}}\Ex^{\Qx}\left[\left(  {\cal L}^{\tilde{\theta}^{\bm \pi}_t}J^{\xi}(Y_t^{{\bm \pi}})-\rho  J^{\xi}(Y_t^{\bm \pi})-q^{\psi}( Y_{t}^{\bm \pi},\tilde{\theta}_{t}^{\bm \pi})\right)^2\right]dt\right)^{\frac{1}{2}}\nonumber\\
&\leq d(\ell(\cdot),\Delta t)\sqrt{\Delta t}\sum_{k=0}^{K-1}\left(\int_{t_k}^{t_{k+1}}\Ex^{\Qx}\left[\left(  {\cal L}^{\tilde{\theta}^{\bm \pi}_t}J^{\xi}(Y_t^{{\bm \pi}})-\rho  J^{\xi}(Y_t^{\bm \pi})-q^{\psi}( Y_{t}^{\bm \pi},\tilde{\theta}_{t}^{\bm \pi})\right)^2\right]dt\right)^{\frac{1}{2}}\nonumber\\
&\leq d(\ell(\cdot),\Delta t)\sqrt{\Delta t}\sqrt{K}\left(\sum_{k=0}^{K-1}\int_{t_k}^{t_{k+1}}\Ex^{\Qx}\left[\left(  {\cal L}^{\tilde{\theta}^{\bm \pi}_t}J^{\xi}(Y_t^{{\bm \pi}})-\rho  J^{\xi}(Y_t^{\bm \pi})-q^{\psi}( Y_{t}^{\bm \pi},\tilde{\theta}_{t}^{\bm \pi})\right)^2\right]dt\right)^{\frac{1}{2}}\nonumber\\
&\leq d(\ell(\cdot),\Delta t)\sqrt{T}\left(\Ex^{\Qx}\left[\int_{0}^{T}\left(  {\cal L}^{\tilde{\theta}^{\bm \pi}_t}J^{\xi}(Y_t^{{\bm \pi}})-\rho  J^{\xi}(Y_t^{\bm \pi})-q^{\psi}( Y_{t}^{\bm \pi},\tilde{\theta}_{t}^{\bm \pi})\right)^2dt\right]\right)^{\frac{1}{2}}\nonumber\\
&\leq d(\ell(\cdot),\Delta t)\sqrt{T}(G^2_J(\xi)+G^2_q(\psi))^{\frac{1}{2}}\Ex^{\Qx}\left[\int_0^{T}\left(\left(\left|(\tilde{\theta}^{\bm \pi}_t)^{\top}\mu \right|+\frac{1}{2}(\tilde{\theta}^{\bm \pi}_t)^{\top}\sigma\sigma^{\top}\tilde{\theta}^{\bm \pi}_t\right.\right.\right.\nonumber\\
&\quad\left.\left.\left.+\left|\sqrt{1-\kappa^2}\sigma_Z(Y_t^{\bm \pi}+1)(\tilde{\theta}^{\bm \pi}_t)^{\top}\sigma\eta\right|+\frac{1}{2}\sigma_Z^2 (Y_t^{\bm \pi}+1)^2+\rho \right) \tilde{J}(Y_t^{\bm \pi})+\tilde{q}( Y_{t}^{\bm \pi},\tilde{\theta}_{t}^{\bm \pi})\right)^2dt\right]^{\frac{1}{2}},
\end{align*}
where $\ell(t):=\varsigma_t$ and $d(\ell(\cdot),\Delta t)=\sup_{|t-s|\leq \Delta t,t,s\in[0,T]}\Ex^{\Qx}[|\ell(t)-\ell(s)|^2]^{\frac{1}{2}}$ is the modulus of continuity of $\ell(\cdot)$ in $\mathbb{L}^2(\Omega)$. Therefore, $d^2(\ell(\cdot),\Delta t)\to 0$ as $\Delta t\to 0$ and we get the desired result \eqref{eq:convergenct-delta}. 
\end{proof}

\begin{proof}[Proof of Theorem \ref{thm:verification}]
By  Theorem \ref{thm:PIT}, it suffices to verify the transversality conditions in \eqref{eq:condition-optimal-pi}. Note that $y\to v(y)$ is non-decreasing,  we have
\begin{align*}
v(y)&=\ln(1+y)+\frac{1}{\rho}\left(-\frac{1}{2}\sigma_Z^2\kappa^2+\sqrt{1-\kappa^2}\zeta+\frac{\gamma}{2}\ln\left(\frac{(2\pi \gamma)^d}{|\sigma\sigma^{\top}|}\right)+\alpha\right)\\
&\geq \frac{1}{\rho}\left(-\frac{1}{2}\sigma_Z^2\kappa^2+\sqrt{1-\kappa^2}\zeta+\frac{\gamma}{2}\ln\left(\frac{(2\pi \gamma)^d}{|\sigma\sigma^{\top}|}\right)+\alpha\right).
\end{align*}
This yields that, for any ${\bm \pi}\in\Pi$,
\begin{align}\label{eq:liminf-v}
\limsup_{T\to \infty}\Ex^{\Qx^W}\left[e^{-\rho T}v(\tilde{Y}^{\bm \pi}_T)\right]\geq 0.
\end{align}

Next, we check the validity of the second transversality condition in \eqref{eq:condition-optimal-pi}, that is
\begin{align}
\lim_{T\to \infty}\Ex^{\Qx^W}\left[e^{-\rho T}v(Y^*_T)\right]=0.\label{eq:transcond}
\end{align}
It follows from \eqref{eq:tilde-Y-pi} and \eqref{eq:optimal-J} that the controlled state process $Y^*=(Y_t^*)_{t \geq 0}$ with policy $\bm{\pi}^*$ obeys the following reflected SDE, for $t> 0$,
\begin{align*}
d Y_t^{*}=(2\alpha+\sqrt{1-\kappa^2}\zeta) (1+Y^*_t) dt+\sqrt{\alpha+\frac{d}{2}\gamma+\frac{1}{2}\kappa^2\sigma_Z^2+\sqrt{1-\kappa^2}\zeta}(1+Y^*_t)dB_t+dL_t^{*}.
\end{align*}
Introduce the processes $H_t:=\ln(1+Y^*_t)$ and $K_t:=\int_0^t \frac{dL_s^*}{1+Y^*_s}$ for $t\geq0$. Then, $t\to H_t$ is a non-negative process wihth $H_0=h=\ln(1+y)$. Moreover, $t\to K_t$ is a non-decreasing process and satisfies $\int_{0}^t{\bf1}_{H_s=0}dK_s=\int_{0}^t{\bf1}_{Y^*_s=0}\frac{dL_s^*}{1+Y^*_s}=\int_0^t\frac{dL_s^*}{1+Y^*_s}=K_t$ a.s., for $t\geq0$. This implies that $K=(K_t)_{t\geq 0}$ is the local time of the process $H=(H_t)_{t\geq 0}$ at the reflecting boundary $0$.  By applying It\^{o}'s formula to $H_t=\ln(1+Y^*_t)$, we arrive at
\begin{align}\label{eq:H}
dH_t&=\hat{\mu}dt+\hat{\sigma}dB_t+\frac{dL_t^{*}}{1+Y_t^*}=\hat{\mu}dt+\hat{\sigma}dB_t+dK_t,\quad t>0,
\end{align}
where constants $\hat{\mu}$ and $\hat{\sigma}$ are defined by
\begin{align*}
\hat{\mu}&:=2\alpha+\sqrt{1-\kappa^2}\zeta-\frac{1}{2}\left(\alpha+\frac{d}{2}\gamma+\frac{1}{2}\kappa^2\sigma_Z^2+\sqrt{1-\kappa^2}\zeta\right),\\
\hat{\sigma}&:=\sqrt{\alpha+\frac{d}{2}\gamma+\frac{1}{2}\kappa^2\sigma_Z^2+\sqrt{1-\kappa^2}\zeta}.
\end{align*}
From the solution representation of the ``Skorokhod problem'', it follows that, for any $y\geq 0$,
\begin{align}\label{eq:K}
K_t&=0 \vee\left\{-h+\max _{s \in[0, t]}\left(-\hat{\mu} s-\hat{\sigma} B_s\right)\right\}
\leq h+|\hat{\mu}|t+|\hat{\sigma}|\max _{s \in[0, t]} B_s.
\end{align}
Thus, we deduce from \eqref{eq:H} and \eqref{eq:K} that, for $T>0$,
\begin{align*}
e^{-\rho T}\Ex^{\Qx^W}[\ln(1+Y_T^*)]=e^{-\rho T}\Ex^{\Qx^W}[K_T]\leq e^{-\rho T}\left(2h+2|\hat{\mu}|T+|\hat{\sigma}|\sqrt{\frac{2T}{\pi}}\right)\to 0, ~\text{as}~T\to \infty.
\end{align*}
This yields the desired transversality condition \eqref{eq:transcond} and completes the proof.
\end{proof}

\noindent
\textbf{Acknowledgements.}\quad {\small The authors are grateful to two anonymous referees for their helpful comments and suggestions. L. Bo and Y. Huang are supported by National Natural Science of Foundation of China (No. 12471451), Natural Science Basic Research Program of Shaanxi (No. 2023-JC-JQ-05) and Shaanxi Fundamental Science Research Project for Mathematics and Physics (No. 23JSZ010). X. Yu is supported by the Hong Kong RGC General Research Fund (GRF) under grant no. 15306523 and grant no. 15211524.}

\end{document}